\documentclass[a4paper,USenglish,11pt,DIV11,final]{scrartcl}

\usepackage{braket,stmaryrd,colonequals,graphicx}
\usepackage{amsmath}
\usepackage{amsthm}
\usepackage{amssymb}
\usepackage{thm-restate}
\usepackage{bm}
\usepackage{booktabs}
\usepackage{xspace}
\usepackage[final]{microtype}
\usepackage{xcolor}
\usepackage{enumitem}
\setlist{itemsep=1pt,topsep=2pt}

\theoremstyle{plain}
\newtheorem{theorem}{Theorem}[section]
\newtheorem{lemma}[theorem]{Lemma}
\newtheorem{corollary}[theorem]{Corollary}
\theoremstyle{definition}
\newtheorem{definition}[theorem]{Definition}

\newtheorem{proposition}[theorem]{Proposition}
\theoremstyle{remark}
\newtheorem*{remark}{Remark}

\newcommand*{\words}[1][\Lambda]{#1{}^{\mathsf{LO}}}
\newcommand*{\terms}[1][\Lambda]{T_{#1}}

\newcommand*{\enc}[1]{\langle #1\rangle}
\newcommand*{\Fragm}{\mathcal{F}}
\newcommand*{\FragmG}{\mathcal{G}}
\newcommand*{\FO}{\mathsf{FO}}

\newcommand*{\Varset}[1][X]{\mathbb{#1}}
\newcommand*{\Vars}{\Varset[V]}
\newcommand*{\Nat}{\mathbb{N}}
\newcommand*{\Rat}{\mathbb{Q}}
\newcommand*{\Int}{\mathbb{Z}}
\newcommand*{\Quantifiers}{\mathcal{Q}}
\newcommand*{\quant}{\mathsf{Q}}

\newcommand*{\dis}[3]{#3/#1 #2}
\newcommand*{\dom}[1]{\operatorname{dom}(#1)}
\newcommand*{\repl}[3]{#1[#2/#3]}
\newcommand*{\fv}{\operatorname{F\kern-0.25emV}}
\newcommand*{\mty}{\operatorname{\mathsf{empty}}}
\renewcommand*{\min}{\operatorname{\mathsf{min}}}
\renewcommand*{\max}{\operatorname{\mathsf{max}}}
\newcommand*{\suc}{\operatorname{\mathsf{suc}}}
\newcommand*{\qd}{\operatorname{qd}}
\newcommand*{\theLO}{\varrho}
\newcommand*{\varLO}{\pi}
\newcommand*{\move}[5]{#1[#2#3,#4,#5]}

\newcommand*{\Dwins}[1]{\mathrel{\lesssim_{#1}}}
\newcommand*{\Swins}[1]{\mathrel{\not\lesssim_{#1}}}
\newcommand*{\Indist}[1]{\mathrel{\approx_{#1}}}

\newcommand*{\evalt}[2]{\llbracket #1\rrbracket_{#2}}
\newcommand*{\DA}{\mathbf{DA}}

\newcommand*{\EF}{Ehrenfeucht-Fra\"iss\'e\xspace}

\renewcommand*{\vec}[1]{\overset\shortrightarrow{\rule{0pt}{0.66ex}\smash{#1}}}
\newcommand*{\cev}[1]{\overset\shortleftarrow{\rule{0pt}{0.66ex}\smash{#1}}}
\newcommand*{\smallvec}[1]{\overset\shortrightarrow{\rule{0pt}{0.33ex}\smash{#1}}}
\newcommand*{\smallcev}[1]{\overset\shortleftarrow{\rule{0pt}{0.33ex}\smash{#1}}}

\newcommand*{\para}[1]{\medskip\par\noindent{\it\normalsize #1}}
\newcommand*{\lpara}[2]{\para{#1:} #2.\quad}
\newcommand*{\basecase}[1]{\lpara{Base case}{#1}}
\newcommand*{\inductivestep}[1]{\lpara{Inductive step}{#1}}

\newcounter{casecounter}
\newcounter{subcasecounter}
\newcommand*{\resetcases}{\setcounter{casecounter}{0}}
\newcommand*{\case}[1]{\stepcounter{casecounter}\setcounter{subcasecounter}{0}\lpara{Case \arabic{casecounter}}{#1}}
\newcommand*{\subcase}[1]{\stepcounter{subcasecounter}\lpara{Case \arabic{casecounter}\alph{subcasecounter}}{#1}}

\newcommand*{\refenum}[1]{{#1.{}}}

\begin{document}

\title{Ehrenfeucht-Fra\"iss\'e Games on Omega-Terms}

\author{Martin Huschenbett\,$^1$ \and
Manfred Kuf\-leitner\,$^2$}
\date{\small $^1$ Institut f\"ur Theoretische Informatik \\ Technische Universit\"at Ilmenau, Germany\\{\texttt{martin.huschenbett@tu-ilmenau.de}} \\ \medskip
$^2$ Institut f\"ur Formale Methoden in der Informatik\footnote{The second author was supported by the
    German Research Foundation (DFG) under grant \mbox{DI 435/5-1}.} \\ Universit\"at Stuttgart, Germany\\{\texttt{kufleitner@fmi.uni-stuttgart.de}}}

\maketitle

\begin{abstract}
  \noindent
  \textsf{\textbf{Abstract.}} \ 
  Fragments of first-order logic over words can often be characterized
  in terms of finite monoids or finite semigroups. Usually these
  algebraic descriptions yield decidability of the question whether a
  given regular language is definable in a particular fragment. An
  effective algebraic characterization can be obtained from identities
  of so-called omega-terms. In order to show that a given fragment
  satisfies some identity of omega-terms, one can use \EF games on
  word instances of the omega-terms. The resulting proofs often
  require a significant amount of book-keeping with respect to the
  constants involved. In this paper we introduce \EF games on
  omega-terms. To this end we assign a labeled linear order to every
  omega-term. Our main theorem shows that a given fragment satisfies
  some identity of omega-terms if and only if Duplicator has a winning
  strategy for the game on the resulting linear orders. This allows to
  avoid the book-keeping.
    
  As an application of our main result, we show that one can decide in
  exponential time whether all aperiodic monoids satisfy some given
  identity of omega-terms, thereby improving a result of McCammond
  (Int. J. Algebra Comput., 2001).
\end{abstract}

\section{Introduction}

By combining a result of McNaughton and Papert~\cite{mp71:short} with
Sch\"utzenberger's characterization of star-free
languages~\cite{sch65sf:short}, a given language over finite words is
definable in first-order logic if and only if its syntactic monoid is
finite and aperiodic. The implication from left to right can be shown
using \EF games, see e.g.~\cite{str94:short}.  A similar result for
two-variable first-order logic $\FO^2$ was obtained by Th\'erien and
Wilke~\cite{tw98stoc:short}: A language is definable in $\FO^2$ if and
only if its syntactic monoid belongs to the variety~$\DA$. Both the
variety $\DA$ and the class of finite aperiodic monoids can be defined
using identities of omega-terms. Roughly speaking, omega-terms are
words equipped with an additional operation, the $\omega$-power.
If~$M$ is a finite monoid, then there exists a positive integer
$\omega_M$ such that $u^{\omega_M} = (u^{\omega_M})^2$ for all $u \in
M$. We call $u^{\omega_M}$ the \emph{idempotent} generated by $u$.
Every mapping $h : \Lambda \to M$ uniquely extends to omega-terms over
$\Lambda$ by setting $h(st) = h(s) h(t)$ and $h(s^\omega) =
h(s)^{\omega_M}$. Now, $M$ satisfies an identity $u = v$ of
omega-terms $u$ and $v$ over the alphabet $\Lambda$ if for every
mapping $h : \Lambda \to M$ we have $h(s) = h(t)$. A finite monoid is
\emph{aperiodic} if and only if it satisfies $a^\omega = a^\omega a$,
and it is in $\DA$ if and only if it satisfies $(abc)^\omega b
(abc)^\omega = (abc)^\omega$, see e.g.~\cite{pin86:short}. Showing
that some first-order fragment $\mathcal{F}$ satisfies an identity $u
= v$ of omega-terms $u,v$ usually works as follows. Suppose
$\mathcal{F}$ does not satisfy $u = v$. Then there exists a formula
$\varphi \in \mathcal{F}$ such that the syntactic monoid of
$L(\varphi)$ does not satisfy $u = v$. The depth $n$ of the formula
$\varphi$ defines an $n$-round \EF game on instances of $u$ and $v$
(i.e., on finite words which are obtained by replacing the
$\omega$-powers by fixed positive integers depending on~$n$). Giving a
winning strategy for Duplicator yields a contradiction, thereby
showing that $\mathcal{F}$ satisfies $u = v$.  Usually, playing the
game on $u$ and $v$ involves some non-trivial book-keeping since one
has to formalize intuitive notions such as positions being near to one
another or being close to some border. For first-order logic and for
$\FO^2$ the book-keeping is still
feasible~\cite{str94:short,dgk08ijfcs:short} whereas for other
fragments such as the quantifier alternation inside $\FO^2$ this task
becomes much more involved (and therefore other techniques are
applied~\cite{KufleitnerLauser13stacs:short,str11csl:short}).

Instead of defining new instances of a given omega-term depending on
the fragment and the number of rounds in the \EF game, we give a
single instance which works for all fragments of first-order logic and
any number of rounds. In addition, we allow an infinite number of
rounds. The fragments we consider in this paper rely on an abstract
notion of logical fragments as introduced
in~\cite{KufleitnerL12icalp:short}. We show that a fragment
$\mathcal{F}$ satisfies an identity of omega-terms if and only if
Duplicator has a winning strategy for the \EF game for $\mathcal{F}$
on the instances of the omega-terms. These instances are labeled
linear orders which, in general, are not finite words.

An obvious application of our main result is the simplification of
proofs showing that some fragment $\mathcal{F}$ satisfies a given
identity of omega-terms. The main reason is that with this new
approach one can avoid the book-keeping. It is slightly less
straightforward that one can use this approach for solving word
problems for omega-terms over varieties of finite monoids. Let
$\mathbf{V}$ be a variety of finite monoids. Then the word problem for
omega-terms over $\mathbf{V}$ is the following: Given two omega-terms
$u$ and $v$, does every monoid in $\mathbf{V}$ satisfy the identity $u
= v$\hspace*{1pt}? This problem was solved for various varieties, see
e.g.~\cite{AlmeidaZ07tcs,McCammond01ijac,Moura11tcs}. Using our main
result, one approach to solving such word problems is as
follows. First, find a logical fragment for $\mathbf{V}$. Second, find
a winning strategy for Duplicator on omega-terms satisfied by this
fragment. Third, use this winning strategy for finding the desired
decision algorithm. In the case of aperiodic monoids, we use this
scheme for improving the decidability result of
McCammond~\cite{McCammond01ijac} by showing that the word problem for
omega-terms over aperiodic monoids is solvable in exponential time.

Historically, the greek letter $\omega$ is used for two different
things which are frequently used throughout this paper: First, the
idempotent power of an element and second, the smallest infinite
ordinal. In order to avoid confusion in our presentation, we chose to
follow the approach of Perrin and Pin~\cite{pp04:short} by using $\pi$
instead of $\omega$ to denote idempotent powers. In particular, we
will use the exponent $\pi$ in omega-terms which is why we will call
them $\pi$-terms in the remainder of this paper.

\section{Preliminaries}

%
As mentioned above, one of the central notions in this paper are so
called $\varLO$-terms. In order to make their interpretation by
several semantics possible in a uniform way, we follow an algebraic
approach. A \emph{$\varLO$-algebra} is a structure
$(U,\,\cdot\,,{}^\varLO)$ comprised of an associative binary operation~$\cdot$ and
a unary operation ${}^\varLO$ on a carrier set $U$. The application of
$\cdot$ is usually written as juxtaposition, i.e., $u v=u\cdot v$, and
the application of ${}^\varLO$ as $u^\varLO$. A
\emph{$\varLO$-term} is an arbitrary element of the free
$\varLO$-algebra $\terms$ generated by $\Lambda$, where
\emph{$\Lambda$ is a countably infinite set which is fixed for the
  rest of this paper}. We also use this set as a universe for letters
(of alphabets).

\subparagraph*{Monoids as $\bm\varLO$-Algebras.}  
Let $M$ be a monoid. For any $k \geq 1$ we extend $M$ to a
$\varLO$-algebra, called \emph{$k$-power algebra on $M$}, by defining
$u^\varLO = u^k$ for $u \in M$. Suppose that $M$ is finite. An element
$u \in M$ is \emph{idempotent} if $u^2 = u$. We extend $M$ to another
$\varLO$-algebra, called \emph{idempotency algebra on $M$}, by
defining $u^\varLO$ for $u \in M$ to be the unique idempotent element
in the set $\set{ u^k | k \geq 1 }$. In fact, there are infinitely
many $k \geq 1$, called \emph{idempotency exponents of $M$}, such that
for each $u \in M$ the element $u^k$ is idempotent, i.e., the
$k$-power algebra and the idempotency algebra on $M$ coincide. An
identity $s=t$ of $\varLO$-terms $s,t \in \terms$ \emph{holds in $M$}
if every $\varLO$-algebra morphism $h$ from $\terms$ into the
idempotency algebra on $M$ satisfies $h(s) = h(t)$.

The set of all \emph{finite words} over an alphabet $A \subseteq
\Lambda$ is~$A^*$. Let $L \subseteq A^*$ be a language over a finite
alphabet $A \subseteq \Lambda$. The \emph{syntactic congruence} of $L$
is the equivalence relation $\equiv_L$ on $A^*$ defined by $u \equiv_L
v$ if $x u y \in L$ is equivalent to $x v y \in L$ for all $x,y\in
A^*$. In fact, $\equiv_L$ is a monoid congruence on~$A^*$. The
quotient monoid $M_L = A^* / {\mathord{\equiv_L}}$ is called
\emph{syntactic monoid} of $L$. It is finite precisely if $L$ is
regular. Suppose that $L$ is regular and let $k \geq 1$ be an
idempotency exponent of $M_L$. Then the map sending each $w \in A^*$
to its $\equiv_L$-class is a $\varLO$-algebra morphism from the
$k$-power algebra on $A^*$ onto the idempotency algebra on
$M_L$. Thus, any identity $s=t$ of $\varLO$-terms $s,t \in \terms$
holds in $M_L$ if, and only if, every $\varLO$-algebra morphism $h$
from $\terms$ into the $k$-power algebra on $A^*$ satisfies $h(s)
\equiv_L h(t)$.

\subparagraph*{Generalised Words.}
The third semantic domain we consider is the class of generalized
words. A \emph{generalized word} (over $\Lambda$) is a triple $u =
(P_u,\leq_u,\ell_u)$ comprised of a (possibly empty) linear ordering
$(P_u,\leq_u)$ being labeled by a map $\ell_u\colon P_u\to
\Lambda$. The set $\dom{u} = P_u$ is the \emph{domain} of $u$, its
elements are called \emph{positions} of~$u$. We write $u(p)$ instead
of $\ell_u(p)$ for $p \in P$. The \emph{order type} of $u$ is
the isomorphism type of $(P_u,\leq_u)$. We regard any finite word $w =
a_1\dotsc a_n \in \Lambda^*$ as a generalized word by defining
$\dom{w}=[1,n]$, $\leq_w$ as the natural order on $[1,n]$ and $w(k) =
a_k$ for $k \in [1,n]$. On that view, generalized words indeed
generalize finite words. As of now, we mean ``generalized word'' when
writing just ``word''. Two words $u$ and $v$ are \emph{isomorphic} if
there exists an isomorphism $f$ of linear orderings from
$(\dom{u},\leq_u)$ to $(\dom{v},\leq_v)$ such that $u(p) = v(f(p))$
for all $p \in \dom{u}$. \emph{We identify isomorphic words.} We
denote the set of all (isomorphism classes of) \emph{countable} words
by~$\words$. The exponent $\mathsf{LO}$ is for \emph{linear order}.
We regard $\Lambda^*$ as a subset of~$\words$.

Let $u,v \in \words$ be two words. Their \emph{concatenation} is the
word $u v \in \words$ defined by $\dom{u v} = \dom{u} \uplus \dom{v}$,
$\leq_{u v}$ makes all positions of $u$ smaller than those of $v$ and
retains the respective orders inside $u$ and inside $v$, and $(u
v)(p)$ is $u(p)$ if $p \in \dom{u}$ and $v(p)$ if $p \in \dom{v}$. The
set $\words$ with concatenation forms a monoid. On finite words this
concatenation coincides with the usual definition and hence
$\Lambda^*$ is a submonoid of $\words$.

It is customary to regard $n \in \Nat$ also as the order type of the
natural linear ordering on $[1,n]$. We extend the notion of the
$n$-power algebra on $\words$ to arbitrary countable order types
$\tau$ as follows. Let $(T,\leq_T)$ be a linear ordering of
isomorphism type $\tau$. The \emph{$\tau$-power} of any word $u \in
\words$ is the word $u^\tau \in \words$ defined by $\dom{u^\tau} =
\dom{u} \times T$, $(p,t) \leq_{u^\tau} (p',t')$ if $t <_T t'$ or if
$t=t'$ and $p \leq_u p'$, and $(u^\tau)(p,t) = u(p)$. We extend the
monoid $\words$ to a $\varLO$-algebra, called \emph{$\tau$-power
  algebra on $\words$}, by defining $u^\varLO = u^\tau$ for $u \in
\words$. We denote by~$\evalt{\,\cdot\,}\tau$ the unique
$\varLO$-algebra morphism from $\terms$ into this $\varLO$-algebra
mapping each $a \in \Lambda$ to the word consisting of a single
position which is labeled by $a$. Finally, notice that there are two
definitions of the $n$-power algebra on $\words$ around, but actually
they coincide.

\subparagraph*{Logic over Words.}
\emph{For the rest of this paper, we fix a countably infinite set
  $\Vars$ of (first-order) variables $x,y,z,\dotsc$.} The syntax of
first-order logic over words is given by
\begin{align*}
	\varphi &\coloncolonequals \top \mid \bot \mid \mty \mid x = y \mid \lambda(x) = a \mid
			x < y \mid x \leq y \mid \suc(x,y) \mid \\ &\mathrel{\phantom{\coloncolonequals}} \min(x) \mid \max(x) \mid 
	 \neg \varphi \mid \varphi \lor \varphi \mid \varphi \land \varphi \mid
			\exists x\, \varphi \mid \forall x\, \varphi \,,
\end{align*}
where $x,y \in \Vars$ and $a \in \Lambda$. The set of all formulae is
denoted by $\FO$. The \emph{free variables} $\fv(\varphi)$ of a
formula $\varphi \in \FO$ are defined as usual. A \emph{sentence} is a
formula $\varphi$ with $\fv(\varphi) = \emptyset$.

We only give a brief sketch of the semantics of formulae. Let $\Varset
\subseteq \Vars$ be a \emph{finite} set of variables. An
\emph{$\Varset$-valuation on $u$} is a pair $\enc{u,\alpha}$
consisting of a word $u \in \words$ and a map $\alpha\colon \Varset
\to \dom{u}$. It is a model of a formula $\varphi \in \FO$ with
$\fv(\varphi) \subseteq \Varset$, in symbols $\enc{u,\alpha} \models
\varphi$, if $u$ satisfies the formula $\varphi$ under the following
assumptions:
\begin{itemize}
\item variables range over positions of $u$ and free variables are
  interpreted according to $\alpha$,
\item $\top$ is always satisfied, $\bot$ never, and $\mty$ only in
  case $\dom{u} = \emptyset$,
\item the function symbol $\lambda$ is interpreted by the labeling
  map $\ell_u\colon \dom{u} \to \Lambda$ and
\item the predicates $<$, $\leq$, $\suc$, $\min$ and $\max$ are
  evaluated in the linear ordering $(\dom{u},\leq_u)$, where
  $\suc(x,y)$ means that $y$ is the immediate successor of~$x$.
\end{itemize}
We identify any word $u \in \words$ with the only
$\emptyset$-valuation on $u$, namely $\enc{u,\emptyset}$ with
$\emptyset$ also denoting the empty map. Thus, for sentences $\varphi$
the meaning of $u \models \varphi$ is well-defined.
Let $A \subseteq \Lambda$ be a finite alphabet and $\varphi \in \FO$ a
sentence. Due to the result of B{\"u}chi, Elgot, and Trakhtenbrot~\cite{Buc60:short,elgot61,tra61:short}, the \emph{language over $A$ defined
  by $\varphi$}, namely $L_A(\varphi) = \set{ w \in A^* | w \models
  \varphi }$, is regular. A language $L \subseteq A^*$ is
\emph{definable in} a class $\Fragm \subseteq \FO$ of formulae if
there exists a sentence $\varphi \in \Fragm$ such that $L =
L_A(\varphi)$.

\subparagraph*{Fragments.} 
We reintroduce (a slight variation of) the notion of a fragment as a
class of formulae obeying natural syntactic closure
properties~\cite{KufleitnerL12icalp:short}. A \emph{context} is a
formula $\mu$ with a unique occurrence of an additional constant
predicate $\circ$ which is intended to be a placeholder for another
formula $\varphi \in \FO$. The result of replacing $\circ$ in $\mu$ by
$\varphi$ is denoted by $\mu(\varphi)$.
Unfortunately, the notion of a fragment as defined in
\cite[Definition~1]{KufleitnerL12icalp:short} is slightly too weak for
our purposes. We require one more \emph{natural} syntactic closure
property, namely condition~\refenum{\ref{cond:negation}} in
Definition~\ref{def:fragment}
below. Condition~\refenum{\ref{cond:forall}} is missing in the
exposition in~\cite{KufleitnerL12icalp:short}. Nevertheless, since we
only add requirements, every fragment in our sense is still a fragment
in the sense of~\cite{KufleitnerL12icalp:short}.

\begin{definition}
\label{def:fragment}
A \emph{fragment} is a non-empty set of formulae $\Fragm \subseteq
\FO$ such that for all contexts $\mu$, formulae $\varphi,\psi \in
\FO$, $a \in \Lambda$ and $x,y \in \Vars$ the following conditions are
satisfied:
\begin{enumerate}
\item If $\mu(\varphi) \in \Fragm$, then $\mu(\top) \in \Fragm$,
  $\mu(\bot) \in \Fragm$, and $\mu(\lambda(x) = a) \in \Fragm$.
\item $\mu(\varphi\lor\psi) \in \Fragm$ if, and only if, $\mu(\varphi)
  \in \Fragm$ and $\mu(\psi) \in \Fragm$.
\item $\mu(\varphi\land\psi) \in \Fragm$ if, and only if,
  $\mu(\varphi) \in \Fragm$ and $\mu(\psi) \in \Fragm$.
\item \label{cond:negation} If $\mu(\neg\neg\varphi) \in \Fragm$, then
  $\mu(\varphi) \in \Fragm$.
\item If $\mu(\exists x\,\varphi) \in \Fragm$ and $x \not\in
  \fv(\varphi)$, then $\mu(\varphi) \in \Fragm$.
\item \label{cond:forall} If $\mu(\forall x\,\varphi) \in \Fragm$ and
  $x \not\in \fv(\varphi)$, then $\mu(\varphi) \in \Fragm$.
\end{enumerate}
It is \emph{closed under negation} if the following condition is
satisfied:
\begin{enumerate}
\setcounter{enumi}{6}
\item If $\varphi \in \Fragm$, then $\neg\varphi \in \Fragm$.
\end{enumerate}
It is \emph{order-stable} if the following condition is satisfied:
\begin{enumerate}
\setcounter{enumi}{7}
\item $\mu(x < y) \in \Fragm$ if, and only if, $\mu(x \leq y) \in
  \Fragm$.
\end{enumerate}
It is \emph{$\suc$-stable} if the following two conditions are
satisfied:
\begin{enumerate}
\setcounter{enumi}{8}
\item If $\mu(\suc(x,y)) \in \Fragm$, then $\mu(x=y) \in \Fragm$,
  $\mu(\max(x)) \in \Fragm$ and $\mu(\min(y)) \in \Fragm$.
\item If $\mu(\min(x)) \in \Fragm$ or $\mu(\max(x)) \in \Fragm$, then
  $\mu(\mty) \in \Fragm$.
\end{enumerate}
\end{definition}

\noindent
Examples for fragments in this sense include all classes of formulae
which are obtained from full first-order logic $\FO$ by limiting the
quantifier depth (e.g., $\FO_n$), the number of quantifier alternations
(e.g., $\Sigma_n$ and $\Pi_n$), the number of quantified variables
(e.g., $\FO^m$), the available predicates (e.g., first-order logic $\FO[<]$ without $\min$, $\max$, $\suc$) or
combinations of those.

The \emph{quantifier depth} $\qd(\varphi)$ of a formula $\varphi \in
\FO$ is defined as usual. A fragment $\Fragm$ has \emph{bounded
  quantifier depth} if there is an $n \in \Nat$ such that $\qd(\varphi)
\leq n$ for all $\varphi \in \Fragm$. For any $n \in \Nat$ and every
fragment $\Fragm$ the set $\Fragm_n = \set{ \varphi \in \Fragm |
  \qd(\varphi) \leq n }$ is a fragment of bounded quantifier
depth. Moreover, the fragment $\Fragm_n$ is order-stable (respectively
$\suc$-stable) in case~$\Fragm$ has the according property.

\section{Ehrenfeucht-Fra\"iss\'e Games for Arbitrary Fragments}
\label{sec:EF_game}

In this section, we introduce an Ehrenfeucht-Fra\"iss\'e game for
arbitrary fragments of first-order logic on generalized words and
develop its basic theory. Before we can describe this game, we need to
define some notation. In the following, we call the ``negated
quantifiers'' $\neg\exists$ and $\neg\forall$ also
\emph{quantifiers}. The set of all quantifiers (in this sense) is
$\Quantifiers = \{\exists,\forall,\neg\exists,\neg\forall\}$. For a
quantifier $\quant \in \Quantifiers$ and a variable $x \in \Vars$, the
\emph{reduct} of $\Fragm$ by $\quant x$ is the set
\begin{equation*}
	\dis\quant x\Fragm = \set{ \varphi \in \FO | \quant x\,\varphi \in \Fragm } \,.
\end{equation*}
Whenever this set is not empty, it is a fragment as well.

Now, let $\Fragm$ be a fragment and $u,v$ two words over $\Lambda$. We
are about to describe the \emph{$\Fragm$-game on $(u,v)$}. A
\emph{configuration} of this game is a triple $S =
(\FragmG,\enc{u,\alpha},\enc{v,\beta})$ comprised of a non-empty,
iterated reduct $\FragmG$ of $\Fragm$ and $\Varset$-valuations
$\enc{u,\alpha}$ and $\enc{v,\beta}$ on $u$ and $v$ for the same
arbitrary \emph{finite} subset $\Varset \subseteq \Vars$. To emphasize
the set $\Varset$, we also speak of an \emph{$\Varset$-configuration}. The
game starts in the $\emptyset$-configuration $(\Fragm,u,v)$ and goes on
for an arbitrary\hspace*{1pt}---\hspace*{1pt}possibly infinite\hspace*{1pt}---\hspace*{1pt}number of rounds. Assuming
that the game is currently in configuration
$S=(\FragmG,\enc{u,\alpha},\enc{v,\beta})$, a single round proceeds as
follows (see~Table~\ref{tab:masterplan} for a summary of this
procedure):
\begin{enumerate}
\item Spoiler chooses a quantifier $\quant \in \Quantifiers$ and a
  variable $x \in \Vars$ such that $\dis\quant x\FragmG \not=
  \emptyset$.
\item Spoiler chooses a position $q$ (like ``quest'') in the domain of
  $u$ if $\quant \in \{\exists,\neg\forall\}$ or in the domain of $v$
  if $\quant \in \{\forall,\neg\exists\}$.
\item Duplicator chooses a position $r$ (like ``reply'') in the domain
  of the other word.
\item The resulting configuration $\move S \quant x q r$ consists of the
  reduct $\dis\quant x\FragmG$ and the extension of the valuations
  $\enc{u,\alpha}$ and $\enc{v,\beta}$ by variable $x$ at positions
  $q$ and $r$, accordingly. Whenever $\quant$ is a negated quantifier,
  the role of the two extended valuations is additionally
  interchanged.
\end{enumerate}
\begin{table}[t]
\centering
\begin{tabular}{cccc}
\toprule
\refenum{1} Spoiler & \refenum{2} Spoiler & \refenum{3} Duplicator & \refenum{4} resulting configuration \\
chooses $\quant x$ & chooses $q$ in & chooses $r$ in & $\move S\quant x q r$ \\ \midrule
$\quant x = \exists x$     & $\dom{u}$ & $\dom{v}$ & $(\dis\exists       x\FragmG,\enc{u,\repl\alpha x q},\enc{v,\repl\beta x r})$ \\
$\quant x = \forall x$     & $\dom{v}$ & $\dom{u}$ & $(\dis\forall       x\FragmG,\enc{u,\repl\alpha x r},\enc{v,\repl\beta x q})$ \\
$\quant x = \neg\exists x$ & $\dom{v}$ & $\dom{u}$ & $(\dis{\neg\exists} x\FragmG,\enc{v,\repl\beta x q},\enc{u,\repl\alpha x r})$ \\
$\quant x = \neg\forall x$ & $\dom{u}$ & $\dom{v}$ & $(\dis{\neg\forall} x\FragmG,\enc{v,\repl\beta x r},\enc{u,\repl\alpha x q})$ \\ \bottomrule
\end{tabular}
\medskip
\caption{A single round of the $\Fragm$-game in configuration $S=(\FragmG,\enc{u,\alpha},\enc{v,\beta})$.}
\label{tab:masterplan}
\end{table}
Whenever a player cannot perform a choice because $\FragmG$ contains
no more quantified formulae or the domain of the according word is
empty, the game immediately stops and the other player wins. 
Besides the inability of Duplicator to move, the winning condition for Spoiler is to reach an $\Varset$-configuration
$(\FragmG,\enc{u,\alpha},\enc{v,\beta})$ such that there exists a literal $\varphi \in \FragmG$ with
$\fv(\varphi) \subseteq \Varset$ and $\enc{u,\alpha} \models \varphi$
but $\enc{v,\beta} \not\models \varphi$; in this case the game immediately stops.
 Duplicator's goal is simply
to prevent Spoiler from winning. In particular, Duplicator wins all
games that go on forever. Due to this circumstance,
the $\Fragm$-game is determined, i.e., either Spoiler or Duplicator has a winning strategy on $(u,v)$.

\begin{remark}
  The $\Fragm$-game is quite asymmetric since Spoiler is not allowed
  to choose before his first move whether he wants to play on $(u,v)$
  or on $(v,u)$. This may lead to the situation that he has a winning
  strategy on $(u,v)$ but not on $(v,u)$ or vice versa. This asymmetry
  is owed to the circumstance that $\Fragm$ might not be closed under
  negation. As soon as $\Fragm$ is assumed to be closed under negation
  this phenomenon disappears and Spoiler has a winning strategy on
  $(u,v)$ if, and only if, he has a winning strategy on
  $(v,u)$. We also note that, in general, the winning condition for Spoiler can be asymmetric since it does not rely on any notion of isomorphism.
  \qed
\end{remark}

\noindent
If the quantifier depth of a fragment $\Fragm$ is bounded by $n \in
\Nat$, the $\Fragm$-game lasts at most $n$ rounds. In particular, for
any fragment $\Fragm$ the $\Fragm_n$-game can be regarded as an
$n$-round version of the $\Fragm$-game. For instance, the $\FO_n$-game
resembles the classical $n$-round Ehrenfeucht-Fra\"iss\'e game. The
following result is an adaption of the Ehrenfeucht-Fra\"iss\'e
Theorem to the~$\Fragm$-game for fragments of bounded quantifier depth.

\pagebreak

\begin{restatable}{theorem}{TheoremEFBounded}
\label{thm:EF_bounded}
Let $\Fragm$ be a fragment of bounded quantifier depth. For all words
$u,v \in \words$ the following are equivalent:%
\begin{enumerate}
\item $u \models \varphi$ implies $v \models \varphi$ for all
  sentences $\varphi \in \Fragm$ and
\item Duplicator has a winning strategy in the $\Fragm$-game on
  $(u,v)$.
\end{enumerate}
\par
\end{restatable}

\noindent
A proof of this theorem can easily be achieved along the lines of a
proof of the classical version, cf.~\cite{ik89iandc}. In fact, such a
proof reveals that the implication ``\refenum{2} $\bm\Rightarrow$
\refenum{1}'' even holds if the quantifier depth of $\Fragm$ is not
bounded. In contrast, the implication ``\refenum{1} $\bm\Rightarrow$
\refenum{2}'' substantially relies the boundedness of the quantifier
depth of $\Fragm$. For instance,
Duplicator has a winning strategy in the $\FO_n$-game on
$(a^\zeta,a^{\zeta+\zeta})$ for each $n \in \Nat$ and hence $a^\zeta
\models \varphi$ implies $a^{\zeta+\zeta} \models \varphi$ for all
sentences $\varphi \in \FO$, but Spoiler has a winning strategy in the infinite
$\FO$-game on~$(a^\zeta,a^{\zeta+\zeta})$.

The objective of the remainder of this section is to identify
additional requirements on $\Fragm$ and/or $u,v$ such that the
boundedness of the quantifier depth can be omitted. It turns out that
the property introduced in Definition~\ref{def:pi-rationality} below
in combination with $\suc$-stability of the fragment is sufficient for
this purpose and still allows for the applications in
Section~\ref{sec:EF_terms}. The order types of the sets $\Nat$,
$\Int$, $\Rat$ and $\Int_{<0}$ ordered naturally are denoted by
$\omega$, $\zeta$, $\eta$ and $\omega^*$, respectively. Then $\omega +
\zeta\cdot\eta + \omega^*$ is the order type of the word $a^\omega
\bigl(a^\zeta\bigr){}^\eta a^{\omega^*}$, where $a \in \Lambda$.

\begin{definition}
  \label{def:pi-rationality}
  Let $\theLO = \omega + \zeta\cdot\eta + \omega^*$. A word $u \in
  \words$ is \emph{$\theLO$-rational} if it can be constructed from
  the finite words in $\words$ using the operations of concatenation
  and $\theLO$-power only or, equivalently, if $u = \evalt t\theLO$
  for some $\varLO$-term $t \in \terms$.
\end{definition}

\begin{restatable}{theorem}{TheoremEFUnbounded}
\label{thm:EF_unbounded}
Let $\Fragm$ be a $\suc$-stable fragment. For all $\theLO$-rational
words $u,v \in \words$ the following are equivalent:
\begin{enumerate}
\item $u \models \varphi$ implies $v \models \varphi$ for all
  sentences $\varphi \in \Fragm$ and
\item Duplicator has a winning strategy in the $\Fragm$-game on
  $(u,v)$.
\end{enumerate}
\par
\end{restatable}

\noindent
As already mentioned, the implication ``\refenum{2} $\bm\Rightarrow$
\refenum{1}'' can be shown using the very same proof as for the
according implication of Theorem~\ref{thm:EF_bounded}. The key idea
behind proving the implication \mbox{``\refenum{1} $\bm\Rightarrow$
  \refenum{2}''} is as follows: Theorem~\ref{thm:EF_bounded} provides
us for each $n \in \Nat$ with a winning strategy for Duplicator in the
$\Fragm_n$-game on $(u,v)$. A winning strategy in the $\Fragm$-game is
obtained by defining a limit of all those strategies. This limit process relies of the
$\theLO$-rationality of the underlying words and is formalized by
Lemma~\ref{lemma:limit_points} below. A major ingredient of its
proof is
Proposition~\ref{prop:concatenation}.

In order to keep notation concise, we abbreviate the circumstance that
Duplicator has a winning strategy in a configuration $S =
(\Fragm,\enc{u,\alpha},\enc{v,\beta})$ by $\enc{u,\alpha} \Dwins\Fragm
\enc{v,\beta}$. Since the $\Fragm$-game is determined, $\enc{u,\alpha}
\Swins\Fragm \enc{v,\beta}$ hence means that Spoiler has a winning
strategy in $S$. The relation $\Dwins\Fragm$ is reflexive and
transitive, i.e., a preorder on the set of all configurations. It induces
an equivalence $\Indist\Fragm$ defined by $\enc{u,\alpha}
\Indist\Fragm \enc{v,\beta}$ if $\enc{u,\alpha} \Dwins\Fragm
\enc{v,\beta}$ and $\enc{v,\beta} \Dwins\Fragm \enc{u,\alpha}$.

\begin{restatable}{proposition}{PropositionConcatenation}
\label{prop:concatenation}
Let $\Fragm$ be a $\suc$-stable fragment, $k \in \Nat$ and
$\enc{u_i,\alpha_i},\enc{v_i,\beta_i}$ $\Varset_i$-valuations with
mutually disjoint $\Varset_i$ for $i \in [1,k]$. If
$\enc{u_i,\alpha_i} \Dwins\Fragm \enc{v_i,\beta_i}$ for each $i \in
[1,k]$, then $\enc{u_1\dotsm u_k,\alpha_1\cup\dotsb\cup\alpha_k}
\Dwins\Fragm \enc{v_1\dotsm v_k,\beta_1\cup\dotsb\cup\beta_k}$.  \qed
\end{restatable}

\begin{restatable}{lemma}{LemmaExponentiation}
\label{lemma:exponentiation}
Let $\Fragm$ be a $\suc$-stable fragment with quantifier depth 
bounded by $n \in \Nat$ and $u,v \in \words$. If $u \Dwins\Fragm v$,
then $u^m \Dwins\Fragm v^\theLO$ and $u^\theLO \Dwins\Fragm v^m$ for
all $m \geq 2^{n+1}-1$.  \qed
\end{restatable}

\noindent
The following lemma formalizes the limit process mentioned above.

\begin{lemma}
\label{lemma:limit_points}
Let $\Fragm$ be a $\suc$-stable fragment, $x \in \Vars$ and
$\enc{u,\alpha}$ an $\Varset$-valuation on a $\theLO$-rational word $u
\in \words$. For every infinite sequence $(q_i)_{i \in \Nat} \in
\dom{u}^\Nat$ there exists a position $q \in \dom{u}$ such that for
all $n \in \Nat$ there are arbitrarily large $i \in \Nat$ with
$\enc{u,\repl\alpha x{q_i}} \Dwins{\Fragm_n} \enc{u,\repl\alpha x q}$.
\end{lemma}

\begin{proof}
  To simplify notation, we call a position $q$ with the property above
  a \emph{$\enc{u,\alpha}$-limit point} of the sequence $(q_i)_{i \in
    \Nat}$ (w.r.t. to $\Fragm$ and $x$). Using this terminology, we
  have to show that every sequence $(q_i)_{i \in \Nat} \in
  \dom{u}^\Nat$ possesses a $\enc{u,\alpha}$-limit point. Since
  neither $\repl\alpha x{q_i}$ nor $\repl\alpha x q$ would depend on
  $\alpha(x)$, we may simply assume that $x \not\in \Varset$. We
  proceed by induction on the $\theLO$-rational construction of $u$.

  \resetcases \basecase{$u$ is finite} Since $\dom{u}$ is finite,
  there exists a $q \in \dom{u}$ such that $q = q_i$ for infinitely
  many $i \in \Nat$. Thus, $q$ is a $\enc{u,\alpha}$-limit point of
  $(q_i)_{i \in \Nat}$.

  \lpara{Inductive step 1}{$u = v_1 v_2$ with $\theLO$-rational
    words $v_1,v_2$} The valuation $\enc{u,\alpha}$ splits into
  valuations $\enc{v_1,\beta_1}$ and $\enc{v_2,\beta_2}$  such that
  $\alpha = \beta_1 \cup \beta_2$. 
  For either $\ell = 1$ or $\ell = 2$ we have $q_i \in \dom{v_\ell}$ for infinitely many $i \in
  \Nat$. Let $I$ be the set of these~$i$. By the induction hypothesis,
  there is a $\enc{v_\ell,\beta_\ell}$-limit point $q \in
  \dom{v_\ell}$ of the subsequence $(q_i)_{i \in
    I}$. Proposition~\ref{prop:concatenation} implies that $q$ is also
  a $\enc{u,\alpha}$-limit point of $(q_i)_{i \in \Nat}$.

  \medskip
  \noindent
  We split the inductive step for $\theLO$-powers in two parts, one for
  $\Varset = \emptyset$ and another for $\Varset \not= \emptyset$.

  \lpara{Inductive step 2}{$u = v^\theLO$ with a $\theLO$-rational $v$
    and $\Varset = \emptyset$} Let $(P,\leq_P)$ be a linear ordering
  of isomorphism type $\theLO$ such that $\dom{u} = \dom{v} \times
  P$. For each $i \in \Nat$ we write $q_i = (s_i,p_i)$. For every $p
  \in P$ let $\cev\tau_p$ and $\vec\tau_p$ be the order types of the
  suborders of $(P,\leq_P)$ induced by the open intervals
  $(-\infty,p)$ and $(p,+\infty)$, respectively. Then $\theLO =
  \cev\tau_p + 1 + \vec\tau_p$. Due to the nature of $\theLO$, each of
  $\cev\tau_p$ and $\vec\tau_p$ is either finite or equals
  $\theLO$. However, the case that $\cev\tau_p$ and $\vec\tau_p$ both
  are finite at the same time cannot occur. Accordingly, we
  distinguish three cases:

  \resetcases 
  \case{$\cev\tau_{p_i} = \vec\tau_{p_i} = \theLO$ for
    infinitely many $i \in \Nat$} Let $I$ be the set of these $i$. By
  the induction hypothesis, there exists a $\enc{v,\emptyset}$-limit
  point $s \in \dom{v}$ of the subsequence $(s_i)_{i \in I}$. We pick
  some $j \in I$. Proposition~\ref{prop:concatenation} reveals that $q
  = (s,p_j)$ is a $\enc{u,\alpha}$-limit point of $(q_i)_{i \in
    \Nat}$.


  \case{$\cev\tau_{p_i}$ is finite and $\vec\tau_{p_i} = \theLO$ for
    infinitely many $i \in \Nat$} Let $I$ be the set of these~$i$. If
  there is an order type which occurs infinitely often among the
  $\cev\tau_{p_i}$ with $i \in I$, the same argumentation as in Case~1
  applies. Henceforth, we assume that such an order type does not
  exist. By the induction hypothesis, the subsequence $(s_i)_{i \in
    I}$ possesses a $\enc{v,\emptyset}$-limit point $s \in
  \dom{v}$. Let $p \in P$ be arbitrary with $\cev\tau_p = \vec\tau_p =
  \theLO$. We show that $q=(s,p)$ is a $\enc{u,\alpha}$-limit point of
  $(q_i)_{i \in \Nat}$.

  Let $n \in \Nat$. Due to the choice of $I$ and $s$, there are
  arbitrarily large $i \in I$ such that $\cev\tau_{p_i}$ is of size at
  least $2^{n+1}-1$ and $\enc{v,\repl\emptyset x{s_i}} \Dwins{\Fragm_n}
  \enc{v,\repl\emptyset x s}$. Lemma~\ref{lemma:exponentiation}
  then implies $v^{\smallcev\tau_{p_i}} \Dwins{\Fragm_n}
  v^\theLO$. Since also $v^{\smallvec\tau_{p_i}} \Dwins{\Fragm_n}
  v^\theLO$, Proposition~\ref{prop:concatenation} yields
  $\enc{u,\repl\emptyset x{q_i}} \Dwins{\Fragm_n}
  \enc{u,\repl\emptyset x q}$.

  \case{$\cev\tau_{p_i} = \theLO$ and $\vec\tau_{p_i}$ is finite for
    infinitely many $i \in \Nat$} Symmetric to Case~2.

  \lpara{Inductive step 3}{$u = v^\theLO$ with a $\theLO$-rational $v$
    and $\Varset \not= \emptyset$} Let $(P,\leq_P)$ be as
  above. Recall that $\Varset$ is supposed to be finite. Let $\tilde
  p_1 <_P \dotsb <_P \tilde p_k$ be an enumeration of all positions $p \in P$
  for which there exists a variable $y \in \Varset$ with $\alpha(y) \in
  \dom{v} \times \{p\}$. We consider the open intervals $P_0 =
  (-\infty,\tilde p_1)$, $P_\ell = (\tilde p_\ell,\tilde p_{\ell+1})$
  for $\ell \in [1,k-1]$, and $P_k = (\tilde p_k,+\infty)$ in
  $(P,\leq_P)$. For $\ell \in [0,k]$ we let $\tau_\ell$ be the order
  type of the suborder induced by $P_\ell$. Then $\theLO = \tau_0 + 1
  + \tau_1 + 1 + \dotsb + 1 + \tau_k$ and hence $u = v^{\tau_0} v
  v^{\tau_1} v \dotsm v v^{\tau_k}$. Due to the nature of $\theLO$,
  each $\tau_\ell$ is either finite or equals $\theLO$. Since for
  every finite $\tau_\ell$ the word $v^{\tau_\ell}$ is the
  concatenation of $\tau_\ell$ copies of $v$, the factorization of $u$
  above is an alternative $\theLO$-rational construction of $u$. This
  construction has the additional property that $\alpha$ does not map
  into the $\theLO$-powers $v^\theLO$. Thus, the induction hypothesis
  and the inductive steps~1 and~2 above yield the claim.
\end{proof}

\noindent
Now, we are prepared to prove the remaining implication of
Theorem~\ref{thm:EF_unbounded}.

\begin{proof}[Proof of Theorem~\ref{thm:EF_unbounded}, ``\refenum{1} $\bm\Rightarrow$ \refenum{2}''.]
  We show that Duplicator can maintain the invariant of staying in
  configurations which are \emph{good} for her. A configuration
  $(\FragmG,\enc{u,\alpha},\enc{v,\beta})$ of the $\Fragm$-game on
  $(u,v)$ is considered to be \emph{good} for Duplicator if
  $\enc{u,\alpha} \Dwins{\FragmG_n} \enc{v,\beta}$ for every $n \in
  \Nat$. Statement~\refenum{1} and Theorem~\ref{thm:EF_bounded} imply
  that the initial configuration $(\Fragm,u,v)$ is good. Moreover, good
  configurations do not meet Spoiler's winning condition as they
  particularly satisfy $\enc{u,\alpha} \Dwins{\FragmG_0}
  \enc{v,\beta}$. Consequently, it suffices to provide a strategy for
  Duplicator which never leaves the set of good configurations since such
  a strategy is a winning strategy.

  Suppose Spoiler chooses the quantifier $\quant x$
  and the quest $q$ in a good configuration
  $(\FragmG,\enc{u,\alpha},\enc{v,\beta})$. We only demonstrate the
  case $\quant = \exists$, where $q \in \dom{u}$. For every $n \in
  \Nat$ we have $\enc{u,\alpha} \Dwins{\FragmG_{n+1}} \enc{v,\beta}$
  and hence there exists $r_n \in \dom{v}$ such that
  $\enc{u,\repl\alpha x q} \Dwins{\dis\exists x{\FragmG_{n+1}}}
  \enc{v,\repl\beta x{r_n}}$. Since $\dis\exists x{\FragmG_{n+1}} =
  (\dis\exists x\FragmG)_n$, this is the same as $\enc{u,\repl\alpha x
    q} \Dwins{(\dis\exists x\FragmG)_n} \enc{v,\repl\beta
    x{r_n}}$. Due to Lemma~\ref{lemma:limit_points}, there exists 
  $r \in \dom{v}$ such that, for every $n \in \Nat$, there are
  infinitely many $i \in \Nat$ with $\enc{v,\repl\alpha x{r_i}}
  \Dwins{(\dis\exists x\FragmG)_n} \enc{v,\repl\alpha x r}$. We show
  that the configuration $\move S\exists x q r = (\dis\exists
  x\FragmG,\enc{u,\repl\alpha x q},\enc{v,\repl\beta x r})$ is good
  again.

  Let $n \in \Nat$. Due to the choice of $r$, there is an $i \geq n$
  with $\enc{v,\repl\alpha x{r_i}} \Dwins{(\dis\exists x\FragmG)_n}
  \enc{v,\repl\alpha x r}$. The position $r_i$ was chosen such that
  $\enc{u,\repl\alpha x q} \Dwins{(\dis\exists x\FragmG)_i}
  \enc{v,\repl\beta x{r_i}}$. Since $n \leq i$, this implies
  $\enc{u,\repl\alpha x q} \Dwins{(\dis\exists x\FragmG)_n}
  \enc{v,\repl\beta x{r_i}}$ and in turn $\enc{u,\repl\alpha x q}
  \Dwins{(\dis\exists x\FragmG)_n} \enc{v,\repl\beta x r}$.
\end{proof}

\section{Ehrenfeucht-Fra\"iss\'e Games on Identities}
\label{sec:EF_terms}

Identities play an important role in the study of the expressive power of first-order fragments.
A recurring problem is to show that a certain identity
of $\varLO$-terms holds in the syntactic monoid/semigroup of every
language definable in the fragment under
consideration. Theorems~\ref{thm:EF_termsM} and~\ref{thm:EF_termsS}
below can remarkably simplify this task, as demonstrated at the end of
this section. In fact, the two theorems are just slight variations of
one another and the sole reason for having two theorems is that the
$\suc$-predicate does not play well with syntactic monoids but only
with syntactic semigroups.

\begin{theorem}
\label{thm:EF_termsM}
Let $\Fragm$ be an order-stable fragment not containing the predicates
$\suc$, $\min$, $\max$ and $\mty$. For all $\varLO$-terms $s,t \in
\terms$ the following are equivalent:
\begin{enumerate}
\item The identity $s=t$ holds in the syntactic monoid of every
  language definable in $\Fragm$.
\item Duplicator has winning strategies in the $\Fragm$-games on
  $(\evalt s\theLO,\evalt t\theLO)$ and $(\evalt t\theLO,\evalt
  s\theLO)$.
\end{enumerate}
\end{theorem}

\pagebreak

\begin{theorem}
\label{thm:EF_termsS}
Let $\Fragm$ be a $\suc$-stable and order-stable fragment. For all
$\varLO$-terms $s,t \in \terms$ the following are equivalent:
\begin{enumerate}
\item The identity $s=t$ holds in the syntactic semigroup of every
  language definable in $\Fragm$ over non-empty words.
\item Duplicator has winning strategies in the $\Fragm$-games on
  $(\evalt s\theLO,\evalt t\theLO)$ and $(\evalt t\theLO,\evalt
  s\theLO)$.\qed
\end{enumerate}
\end{theorem}

\noindent
The main ingredients of the proofs of both theorems are
Theorem~\ref{thm:EF_unbounded} and
\mbox{\cite[Proposition~2]{KufleitnerL12icalp:short}} which is
restated as Proposition~\ref{prop:inverse_morphisms} below.

\begin{proposition}
\label{prop:inverse_morphisms}
Let $\Fragm$ be a fragment, $A,B\subseteq\Lambda$ finite alphabets and
$h$ a monoid morphism from $A^*$ into $B^*$. Suppose the following:
\begin{enumerate}
\item If $\Fragm$ contains the predicate $\leq$ or $<$, then $\Fragm$
  is order-stable or $h(A) \subseteq B\cup\{\varepsilon\}$.
\item If $\Fragm$ contains the predicate $\suc$, $\min$, $\max$ or
  $\mty$, then $\varepsilon \not\in h(A)$.
\end{enumerate}
Then  $h^{-1}(L)$ is $\Fragm$-definable whenever $L \subseteq A^*$ is $\Fragm$-definable.
\end{proposition}

\noindent
Applying this proposition to $\Fragm$-games yields that monoid
morphisms satisfying the two conditions above preserve the existence
of winning strategies for Duplicator.

\begin{corollary}
\label{cor:substitution}
Let $\Fragm$, $A$, $B$ and $h$ be as in
Proposition~\ref{prop:inverse_morphisms} satisfying
conditions~\refenum{1} and~\refenum{2}. Then $u \Dwins\Fragm v$
implies $h(u) \Dwins\Fragm h(v)$ for all $u,v \in A^*$.
\end{corollary}

\begin{proof}
  Let $u,v \in A^*$ with $u \Dwins\Fragm v$. Since finite words are
  $\theLO$-rational and due to Theorem~\ref{thm:EF_unbounded}, it
  suffices to show that $h(u) \models \varphi$ implies $h(v) \models
  \varphi$ for all sentences $\varphi \in \Fragm$. Consider a sentence
  $\varphi \in \Fragm$. By Proposition~\ref{prop:inverse_morphisms},
  there is a sentence $\psi \in \Fragm$ such that $L_A(\psi) =
  h^{-1}\bigl(L_B(\varphi)\bigr)$. Altogether, $h(u) \models \varphi$
  implies $u \models \psi$ and since $u \Dwins\Fragm v$ this implies
  $v \models \psi$ which in turn implies $h(v) \models \varphi$.
\end{proof}

\noindent
The following corollary is an immediate consequence of
Proposition~\ref{prop:concatenation} and
Lemma~\ref{lemma:exponentiation}.

\begin{restatable}{corollary}{CorollaryTermInterpretations}
\label{cor:term_interpretations}
Let $\Fragm$ be a $\suc$-stable fragment whose quantifier depth is
bounded by $n \in \Nat$ and let $t \in \terms$ be a
$\varLO$-term. Then $\evalt t\theLO \Indist\Fragm \evalt t m$ for all
$m \geq 2^{n+1}-1$.  \qed
\end{restatable}

\noindent
The previous results allow us to show Theorems~\ref{thm:EF_termsM}
and~\ref{thm:EF_termsS}. However, since their proofs are as similar as
their statements, we only demonstrate the first one.

\begin{proof}[Proof of Theorem~\ref{thm:EF_termsM}]
  Let $A \subseteq \Lambda$ be the finite set containing all $u \in
  \Lambda$ appearing in $s$ or~$t$. We show both implications
  separately.

\para{``\refenum{1} $\bm\Rightarrow$ \refenum{2}''.}
By Theorem~\ref{thm:EF_unbounded}, it suffices to show for every
sentence $\varphi \in \Fragm$ that $\evalt s\theLO \models \varphi$
if and only if $\evalt t\theLO \models \varphi$. Consider a sentence $\varphi \in
\Fragm$ and put $n = \qd(\varphi)$. We put $L = L_A(\varphi)$ and let
$k \geq 2^{n+1}-1$ be an idempotency exponent of $M_L$. We consider an
arbitrary $\varLO$-algebra morphism $h$ from $T_A$ into the
$k$-power algebra on $A^*$ with $h(a) = a$ for each $a \in A$. Because
$s=t$ holds in $M_L$, we have $h(s) \equiv_L h(t)$. Since $h(s) =
\evalt s k$ as well as $h(t) = \evalt t k$ and by
Corollary~\ref{cor:term_interpretations}, we obtain $h(s)
\Indist{\Fragm_n} \evalt s\theLO$ and $h(t) \Indist{\Fragm_n} \evalt
t\theLO$. Altogether, we conclude that $\evalt s\theLO \models
\varphi$ if and only if $h(s) \models \varphi$ if and only if $h(t) \models \varphi$ if and only if
$\evalt t\theLO \models \varphi$.

\para{``\refenum{2} $\bm\Rightarrow$ \refenum{1}''.}
Let $B \subseteq \Lambda$ be a finite alphabet and $L \subseteq B^*$ a
language defined by a sentence $\varphi \in \Fragm$. Let $n =
\qd(\varphi)$ and $k \geq 2^{n+1}-1$ be an idempotency exponent of
$M_L$. We have to show that every $\varLO$-algebra morphism $g$ from
$T_A$ into the $k$-power algebra on $B^*$ satisfies $g(s) \equiv_L
g(t)$. Consider such a morphism $g$ and let $h$ be the unique monoid
morphism from $A^*$ into $B^*$ defined by $h(a) = g(a)$ for each $a \in
A$. Then $g(s) = h(\evalt s k)$ and $g(t) = h(\evalt t
k)$. Corollary~\ref{cor:term_interpretations} and the assumption
$\evalt s\theLO \Indist\Fragm \evalt t\theLO$ yield $\evalt s k
\Indist{\Fragm_n} \evalt s\theLO \Indist{\Fragm_n} \evalt t\theLO
\Indist{\Fragm_n} \evalt t k$. We conclude $g(s) \Indist{\Fragm_n}
g(t)$ by Corollary~\ref{cor:substitution}. By
Proposition~\ref{prop:concatenation}, we obtain $u g(s) v
\Indist{\Fragm_n} u g(t) v$ for all $u,v \in B^*$. Since $\varphi \in
\Fragm_n$, this finally implies $g(s) \equiv_L g(t)$.
\end{proof}

\noindent
In the remainder of this section, we demonstrate two applications of
Theorem~\ref{thm:EF_termsM} by providing quite short proofs of two
well-known results. The following corollary can be obtained by combining a result of McNaughton and Papert~\cite{mp71:short} with Sch\"utzenberger's characterization of star-free languages~\cite{sch65sf:short}. A more direct proof can, for instance, be found in~\cite{str94:short}. A finite monoid $M$ is called \emph{aperiodic} if
the identity $a^\varLO a = a^\varLO$ holds in $M$.

\begin{corollary}
The syntactic monoid of every first-order definable language is aperiodic.
\end{corollary}

\begin{proof}
  The
  predicates $\suc$, $\min$, $\max$ and $\mty$ can be expressed in $\FO[<]$. By Theorem~\ref{thm:EF_termsM}, it suffices to show $\evalt{a^\varLO a}\theLO
  \Indist{\FO[<]} \evalt{a^\varLO}\theLO$. The property $\theLO + 1 =
  \theLO$ of the order type $\theLO$ implies $\evalt{a^\varLO a}\theLO
  = \evalt{a^\varLO}\theLO$ and the claim follows.
\end{proof}

\noindent
The second application relates definability in
$\FO^2[<]$ to the class $\DA$. The fragment $\FO^2[<]$ consists of all
formulae not containing the predicates $\suc$, $\min$, $\max$ and
$\mty$ which quantify over two fixed variables $x_1,x_2 \in \Vars$
only. The class $\DA$ consists of all finite monoids in which the
identity $(a b c)^\varLO b (a b c)^\varLO = (a b c)^\varLO$ holds. A
significant amount of book-keeping is involved when showing that the
syntactic monoid of every $\FO^2[<]$-definable language is in $\DA$ by
applying the classical \EF game approach, see
e.g.~\cite{dgk08ijfcs:short}\footnote{Actually, the proof given
  in~\cite{dgk08ijfcs:short} does not use the language of \EF games,
  but it can easily be restated this way.}. On the other hand, the
abstract idea of this proof is very simple: Duplicator copies every
move near the left and near the right border, and he does not need to
care in the center.  We now show that this idea can easily be
formalized when using Theorem~\ref{thm:EF_termsM}.

\begin{corollary}
\label{cor:FO2}
The syntactic monoid of any language definable in $\FO^2[<]$ is in
$\DA$.
\end{corollary}

\begin{proof}
  Let $s = (a b c)^\varLO b (a b c)^\varLO$ and $t = (a b
  c)^\varLO$. Again by Theorem~\ref{thm:EF_termsM}, it suffices to
  show $\evalt s\theLO \Indist{\FO^2[<]} \evalt t\theLO$. With $u = (a
  b c)^{\omega} (a b c)^{\zeta\cdot\eta} $ and $v = (a b c)^{\zeta\cdot\eta} (a b c)^{\omega^*}$ we obtain
\begin{equation*}
	\evalt s\theLO = u\hspace*{1pt}(a b c)^{\omega^*} b (a b c)^{\omega} \hspace*{1pt} v
	\qquad\text{and}\qquad
	\evalt t\theLO = u \hspace*{1pt} (a b c)^{\omega^*} (a b c)^{\omega} \hspace*{1pt} v \,.
\end{equation*}
Since $\FO^2[<]$ is closed under negation and due to
Proposition~\ref{prop:concatenation}, it further suffices to show that
Duplicator has a winning strategy in the $\FO^2[<]$-game on
\begin{equation*}
	\bigl((a b c)^{\omega^*} b (a b c)^{\omega},\, (a b c)^{\omega^*} (a b c)^{\omega} \bigr) \,.
\end{equation*}
The strategy is to choose a reply that is labeled by the same letter
as the request and such that the positions corresponding to $x_1$ and
$x_2$ are in the same order in both words. This is always possible,
since in both words there are always infinitely many positions to the
left (respectively to the right) of any position which are labeled by
a given letter from $a,b,c$.
\end{proof}

\section{The Word Problem for $\pi$-Terms over Aperiodic Monoids}
\label{sec:decidability_result}

The word problem for $\pi$-terms over aperiodic monoids was solved by
McCammond~\cite{McCammond01ijac} by computing normal forms. In the
process of computing these normal forms the intermediate terms can
grow and, to the best of our knowledge, neither the worst-case running
time nor the maximal size of the intermediate terms has been estimated
(and it seems to be difficult to obtain such results). In this section
we give an exponential algorithm for solving the word problem for
$\pi$-terms over aperiodic monoids. Our algorithm does not compute
normal forms as $\pi$-terms; instead we show that the evaluation under
$\evalt{\,\cdot\,}{\theLO}$ can be used as a normal form for
$\pi$-terms.

\begin{theorem}
\label{thm:equation_theory}
Given two $\varLO$-terms $s,t \in \terms$, one can decide whether the
identity $s=t$ holds in every aperiodic monoid in time exponential in
the size of $s$ and $t$.
\end{theorem}

\noindent
The proof is a reduction to the isomorphism problem for regular
words, cf.~\cite{BloomEsik05}. These generalized words particularly
include all $\theLO$-rational words and can be described by
expressions similar to $\varLO$-terms but using $\omega$-power,
$\omega^*$-power and dense shuffle instead of the $\varLO$-power. Due
to \cite[Theorem~79]{BloomEsik05}, one can decide in polynomial time
whether two such expressions describe isomorphic words. The
characterization underlying the reduction is as follows:

\begin{proposition}
\label{prop:reduction}
For all $\varLO$-terms $s,t \in \terms$ the following conditions are equivalent:
\begin{enumerate}
\item The identity $s = t$ holds in every aperiodic finite monoid.
\item $\evalt s\theLO = \evalt t\theLO$.
\end{enumerate}
\end{proposition}

\begin{proof}
\para{``\refenum{1} $\bm\Rightarrow$ \refenum{2}''.}
The results in \cite{McCammond01ijac} imply that the identity $s = t$
can be deduced from the following list of axioms, where $n \geq 1$:
\begin{align*}
	\label{eq:axioms}
	(u v) w &= u (v w) &
	(u^\varLO)^\varLO &= u^\varLO &
	(u^n)^\varLO &= u^\varLO \\
	u^\varLO u^\varLO &= u^\varLO &
	u^\varLO u &= u u^\varLO = u^\varLO &
	(u v)^\varLO u &= u (v u)^\varLO \, .
\end{align*}
As a matter of fact, the $\theLO$-power algebra on $\words$ satisfies
these axioms as well. Consequently, $\evalt s\theLO = \evalt t\theLO$
can be proved along a deduction of the identity $s = t$ from the axioms.

\para{``\refenum{2} $\bm\Rightarrow$ \refenum{1}''.}
Due to Eilenberg's Variety Theorem~\cite{eil76:short}, the
pseudovariety of aperiodic monoids is generated by the class of
syntactic aperiodic monoids. The latter are precisely the syntactic
monoids of first-order definable
languages~\cite{mp71:short,sch65sf:short}. By Theorem~\ref{thm:EF_termsM} the identity $s=t$ holds in the syntactic monoid of every such language.
\end{proof}

\begin{proof}[Proof of Theorem~\ref{thm:equation_theory}]
  In order to apply the decision procedure from
  \cite[Theorem~79]{BloomEsik05}, we have to translate $s$ and $t$
  into expressions generating the same words and which do not use
  $\theLO$-power but $\omega$-power, $\omega^*$-power and dense
  shuffle instead. Such a translations can be based on the identity
  $u^\theLO = u^\omega \bigl( u^{\omega^*} u^\omega \bigr)^\eta
  u^{\omega^*}$ which holds for all words $u \in \words$. Therein, the
  $\eta$-power is a special case of the dense shuffle. Since this
  translation leads to a blow-up which is exponential in the number of
  nested applications of $\varLO$-powers within $s$ and $t$, we can decide
  $\evalt s\theLO = \evalt t\theLO$ in time at most exponential in the
  size of $s$ and $t$.
\end{proof}

\section{Summary}

For every $\pi$-term $t$ we define a labeled linear order $\evalt
t\theLO$, and every first-order fragment~$\mathcal{F}$ over finite
words naturally yields a (possibly infinite) \EF game on labeled
linear orders. The important property of these constructions is that
$\mathcal{F}$ satisfies an identity $s = t$ of $\pi$-terms $s$ and $t$
if and only if Duplicator has a winning strategy in the
$\mathcal{F}$-game on $\evalt s\theLO$ and $\evalt t\theLO$. 
We note that $\evalt t\theLO$ does not depend on
$\mathcal{F}$.
Usually
showing that a fragment~$\mathcal{F}$ satisfies an identity $s = t$
requires a significant amount of book-keeping which in most cases is
not part of the actual proof idea. Our main results
Theorem~\ref{thm:EF_termsM} and Theorem~\ref{thm:EF_termsS} allow to
formalize such proof ideas without further book-keeping, see
e.g.~Corollary~\ref{cor:FO2}. A probably less obvious application of
our main result are word problems for $\pi$-terms over varieties of
finite monoids. We show that the word problem for $\pi$-terms over
aperiodic finite monoids is solvable in exponential time
(Theorem~\ref{thm:equation_theory}), thereby improving a result of
McCammond~\cite{McCammond01ijac}.
Theorems~\ref{thm:EF_termsM} and~\ref{thm:EF_termsM} can easily be strengthened by using ordered identities and ordered syntactic monoids. We refrained to do so in order to avoid additional technicalities.

\enlargethispage{\baselineskip} 


{\small
\newcommand{\Ju}{Ju}\newcommand{\Ph}{Ph}\newcommand{\Th}{Th}\newcommand{\Ch}{Ch}\newcommand{\Yu}{Yu}\newcommand{\Zh}{Zh}\newcommand{\St}{St}
}

\clearpage

\begin{appendix}

\renewcommand*{\thetheorem}{\Alph{section}.\arabic{theorem}}
\renewcommand*{\theproposition}{\Alph{section}.\arabic{theorem}}

\section{Missing Proofs from Section~\ref{sec:EF_game}}

\subsection{Missing Proofs for Theorem~\ref{thm:EF_bounded}}

{
\renewcommand*{\thetheorem}{\ref{thm:EF_bounded}}
\TheoremEFBounded*
}

\noindent
As indicated earlier, the implication \mbox{``\refenum{2}
  $\bm\Rightarrow$ \refenum{1}''} even holds if the quantifier depth of
$\Fragm$ is not bounded. The contraposition of this claim is shown by
the proposition below.

\begin{proposition}
\label{prop:Spoiler_wins}
Let $\Fragm$ be a fragment and $\enc{u,\alpha},\enc{v,\beta}$
$\Varset$-valuations. If there exists a formula $\varphi \in \Fragm$
with $\fv(\varphi) \subseteq \Varset$ such that $\enc{u,\alpha}
\models \varphi$ but $\enc{v,\beta} \not\models \varphi$, then
$\enc{u,\alpha} \Swins\Fragm \enc{v,\beta}$.
\end{proposition}

\begin{proof}
  Using De Morgan's laws, we can transform $\varphi$ into a
  \emph{positive} Boolean combination of literals and formulae
  starting with one of the quantifiers in $\Quantifiers$. Due to the
  syntactic closure properties of fragments, the resulting formula
  also belongs to $\Fragm$. Thus, we may assume without loss of
  generality that $\varphi$ itself is a literal or a quantified
  formula. We proceed by induction on $\qd(\varphi)$.

  \basecase{$\qd(\varphi)=0$, i.e., $\varphi$ is a literal} The
  assumption that $\enc{u,\alpha} \models \varphi$ but $\enc{v,\beta}
  \not\models \varphi$ means that the configuration
  $(\Fragm,\enc{u,\alpha},\enc{v,\beta})$ satisfies Spoilers winning
  condition.

  \inductivestep{$\qd(\varphi)>0$, i.e., $\varphi$ is a quantified
    formula} Let $\varphi = \quant x\,\psi$ with $\quant \in
  \Quantifiers$, $x \in \Vars$ and $\psi \in \dis\quant x\Fragm$. We
  only demonstrate the case $\quant = \exists$. Since $\enc{u,\alpha}
  \models \varphi$, there exists a $q \in \dom{u}$ such that
  $\enc{u,\repl\alpha x q} \models \psi$. Let Spoiler choose
  quantifier $\exists x$ and quest $q$ in configuration
  $(\Fragm,\enc{u,\alpha},\enc{v,\beta})$. If $\dom{v} = \emptyset$,
  he immediately wins the game. Otherwise, suppose that Duplicator
  replies by $r \in \dom{v}$. Since $\enc{v,\beta} \not\models
  \varphi$, we obtain $\enc{v,\repl\beta x r} \not\models \psi$. By
  the induction hypothesis, this implies $\enc{u,\repl\alpha x q}
  \Swins{\dis\exists x\Fragm} \enc{v,\repl\beta x r}$ and in turn
  $\enc{u,\alpha} \Swins\Fragm \enc{v,\beta}$.
\end{proof}

\noindent
The following proposition is a well-known fact about first-order logic.

\begin{proposition}
\label{proposition:finite_repr_FO}
Let $n \in \Nat$ and $\Varset \subseteq \Vars$ be a finite set of
variables. There exists a finite subset $\Phi \subseteq \FO_n$ of
formulae $\varphi$ with $\fv(\varphi) \subseteq \Varset$ such that
every formula $\psi \in \FO_n$ with $\fv(\psi) \subseteq \Varset$ is
equivalent to a formula in $\Phi$.
\end{proposition}

\begin{corollary}
\label{cor:finite_repr}
Let $\Fragm$ be a fragment of bounded quantifier depth and $\Varset
\subseteq \Vars$ a finite set of variables. There exists a finite
subset $\Phi \subseteq \Fragm$ of formulae $\varphi$ with
$\fv(\varphi) \subseteq \Varset$ such that every formula $\psi \in
\Fragm$ with $\fv(\psi) \subseteq \Varset$ is equivalent to a formula
in $\Phi$.
\end{corollary}

\begin{proof}
  Let $n \in \Nat$ be a bound on the quantifier depth of $\Fragm$ and
  $\Psi \subseteq \FO_n$ one of the finite sets whose existence is
  guaranteed by Proposition~\ref{proposition:finite_repr_FO}. Let $\Xi \subseteq
  \Psi$ be the set of those $\xi \in \Psi$ for which there exists an
  equivalent formula $\varphi_\xi \in \Fragm$ with $\fv(\varphi_\xi)
  \subseteq \Varset$. We show the finite set $\Phi = \set{ \varphi_\xi
    | \xi \in \Xi }$ to have the desired property. Consider an
  arbitrary formula $\psi \in \Fragm$ with $\fv(\psi) \subseteq
  \Varset$. By choice of $\Psi$, there exists a formula $\xi
  \in \Psi$ which is equivalent to $\psi$. In particular, $\xi \in
  \Xi$. Thus, the formula $\varphi_\xi \in \Phi$ is equivalent to
  $\psi$ as well.
\end{proof}

\begin{proposition}
\label{prop:Duplicator_wins}
Let $\Fragm$ be a fragment of bounded quantifier depth and let
$\enc{u,\alpha},\enc{v,\beta}$ be $\Varset$-valuations. If
$\enc{u,\alpha} \models \varphi$ implies $\enc{v,\beta} \models
\varphi$ for all formulae $\varphi \in \Fragm$ with $\fv(\varphi)
\subseteq \Varset$, then $\enc{u,\alpha} \Dwins\Fragm \enc{v,\beta}$.
\end{proposition}

\begin{proof}
  Let $n \in \Nat$ be a bound on the quantifier depth of $\Fragm$. We
  proceed by induction on $n$.

  \basecase{$n=0$} By assumption, the configuration
  $(\Fragm,\enc{u,\alpha},\enc{v,\beta})$ does not satisfy Spoiler's
  winning condition. Since $\Fragm$ contains no quantified formulae,
  this implies $\enc{u,\alpha} \Dwins\Fragm \enc{v,\beta}$.

  \inductivestep{$n>0$} Suppose that Spoiler chooses quantifier
  $\quant x$ and quest $q$ in configuration
  $(\Fragm,\enc{u,\alpha},\enc{w,\gamma})$. We are looking for a
  suitable reply $r$ for Duplicator. We only demonstrate the case
  $\quant = \exists$, where $q \in \dom{u}$. By
  Corollary~\ref{cor:finite_repr}, there exists a finite set $\Phi
  \subseteq \dis\exists x\Fragm$ of formulae $\varphi$ with
  $\fv(\varphi) \subseteq \Varset\cup\{x\}$ such that every formula
  $\psi \in \dis\exists x\Fragm$ with $\fv(\psi) \subseteq
  \Varset\cup\{x\}$ is equivalent to a formula in $\Phi$. We consider
  the formulae
\begin{equation*}
	\xi = \bigwedge\nolimits_{\varphi\in\Phi, \enc{u,\repl\alpha x q} \models \varphi} \varphi
	\qquad\text{and}\qquad
	\chi = \exists x\,\xi \,.
\end{equation*}
Notice that $\xi \in \dis\exists x\Fragm$ and $\chi \in \Fragm$. Due
to the choice of $\xi$ and $\chi$, we have $\enc{u,\repl\alpha x q}
\models \xi$ and hence $\enc{u,\alpha} \models \chi$. Using the
assumption of the proposition, we conclude $\enc{v,\beta} \models
\chi$. Thus, there is an $r \in \dom{v}$ such that $\enc{v,\repl\beta
  x r} \models \xi$. We show that $r$ is a suitable reply for
Duplicator, i.e., $r$ satisfies $\enc{u,\repl\alpha x q}
\Dwins{\dis\exists x\Fragm} \enc{v,\repl\beta x r}$. By the induction
hypothesis and due to the choice of $\Phi$, it suffices to show that
all formulae $\varphi \in \Phi$ with $\enc{u,\repl\alpha x q} \models
\varphi$ also satisfy $\enc{v,\repl\beta x r} \models \varphi$. This follows immediately from the choice of $\xi$ and
$\enc{v,\repl\beta x r} \models \xi$.
\end{proof}

\subsection{Missing Proofs for Theorem~\ref{thm:EF_unbounded}}
\label{apx:EF_unbounded}

{
\renewcommand*{\thetheorem}{\ref{thm:EF_unbounded}}
\TheoremEFUnbounded*
}

\noindent
In the course of proving this theorem, we have to explicitly
construct winning strategies for Duplicator several times. The general
pattern behind the constructions always consists of the following four
steps (cf. the proof of Theorem~\ref{thm:EF_unbounded}):
\begin{enumerate}
\item We declare a class of configurations of the game which we call
  \emph{good} for Duplicator.
\item We show that the initial configuration of the game is good.
\item We show that no good configuration satisfies the winning condition
  of Spoiler.
\item We show that for every good configuration and any choice of a
  quantifier and a quest by Spoiler, there exists a reply by
  Duplicator such that the resulting configuration is again good.
\end{enumerate}
Whenever these four steps can be implemented, Duplicators winning
strategy is to simply stay within the class of good
configurations.

\begin{proposition}
\label{prop:transitivity}
Let $\Fragm$ be a fragment. The relation $\Dwins\Fragm$ is transitive.
\end{proposition}

\begin{proof}
  Consider $\Varset$-valuations
  $\enc{u,\alpha},\enc{v,\beta},\enc{w,\gamma}$ with $\enc{u,\alpha}
  \Dwins\Fragm \enc{v,\beta}$ and $\enc{v,\beta} \Dwins\Fragm
  \enc{w,\gamma}$. We have to show $\enc{u,\alpha} \Dwins\Fragm
  \enc{w,\gamma}$. For this purpose, we adhere closely to the four
  steps of constructing a winning strategy for Duplicator.
\begin{enumerate}
\item We call a $\Varset[Y]$-configuration
  $(\FragmG,\enc{u,\alpha'},\enc{w,\gamma'})$ \emph{good} if there
  exists a $\Varset[Y]$-valuation $\enc{v,\beta'}$ such that
  $\enc{u,\alpha'} \Dwins\FragmG \enc{v,\beta'}$ and $\enc{v,\beta'}
  \Dwins\FragmG \enc{w,\gamma'}$.
\item Clearly, the configuration $(\Fragm,\enc{u,\alpha},\enc{w,\gamma})$
  is good.
\item Let $S=(\FragmG,\enc{u,\alpha'},\enc{w,\gamma'})$ be a good
  $\Varset[Y]$-configuration and $\enc{v,\beta'}$ a witnessing
  $\Varset[Y]$-valuation. For every literal $\varphi \in \FragmG$ with
  $\fv(\varphi) \subseteq \Varset[Y]$ and $\enc{u,\alpha'} \models
  \varphi$, the choice of $\beta'$ implies $\enc{v,\beta'} \models
  \varphi$ and in turn $\enc{w,\gamma'} \models \varphi$. Thus, $S$
  does not satisfy Spoiler's winning condition.
\item Suppose that Spoiler chooses quantifier $\quant x$ and quest $q$
  in a good configuration $S =
  (\FragmG,\enc{u,\alpha'},\enc{w,\gamma'})$. We are looking for a
  suitable reply $r$ for Duplicator. We only demonstrate the case
  $\quant = \exists$, where $q \in \dom{u}$. Let $\enc{v,\beta'}$ be a
  valuation witnessing that $S$ is good. Since $\enc{u,\alpha'}
  \Dwins\FragmG \enc{v,\beta'}$ and $\enc{v,\beta'} \Dwins\FragmG
  \enc{w,\gamma'}$, there exist a $p \in \dom{v}$ such that
  $\enc{u,\repl{\alpha'} x q} \Dwins{\dis\exists x\FragmG}
  \enc{v,\repl{\beta'} x p}$ and an $r \in \dom{w}$ such that
  $\enc{v,\repl{\beta'} x p} \Dwins{\dis\exists x\FragmG}
  \enc{w,\repl{\gamma'} x r}$. Thus, $\move S\exists x q r$ is a good
  configuration.  \qedhere
\end{enumerate}
\end{proof}

\noindent
The following two lemmas are auxiliary statements for the proof of
Proposition~\ref{prop:concatenation}.

\begin{lemma}
\label{lemma:discard_vars}
Let $\Fragm$ be a fragment, $\enc{u,\alpha},\enc{v,\beta}$
$\Varset$-valuations and $\enc{u,\alpha'},\enc{v,\beta'}$ their
respective restrictions to $\Varset'$-valuations for some subset
$\Varset'\subseteq\Varset$. If $\enc{u,\alpha} \Dwins\Fragm
\enc{v,\beta}$, then $\enc{u,\alpha'} \Dwins\Fragm \enc{v,\beta'}$.
\end{lemma}

\begin{proof}
  We adhere closely to the four steps of constructing a winning
  strategy for Duplicator.
\begin{enumerate}
\item A $\Varset[Y]'$-configuration
  $(\FragmG,\enc{u,\gamma'},\enc{v,\delta'})$ is \emph{good} if there
  exist a finite set $\Varset[Y] \subseteq \Vars$ with $\Varset[Y]
  \supseteq \Varset[Y]'$ and $\Varset[Y]'$-valuations
  $\enc{u,\gamma},\enc{v,\delta}$ extending
  $\enc{u,\gamma'},\enc{v,\delta'}$ such that $\enc{u,\gamma}
  \Dwins\Fragm \enc{v,\delta}$.
\item Clearly, the configuration $(\Fragm,\enc{u,\alpha'},\enc{v,\beta'})$
  is good.
\item Let $S = (\FragmG,\enc{u,\gamma'},\enc{v,\delta'})$ be a good
  $\Varset[Y]'$-configuration and $\enc{u,\gamma},\enc{v,\delta}$
  witnessing $\Varset[Y]$-valuations. Consider a literal $\varphi \in
  \FragmG$ with $\fv(\varphi) \subseteq \Varset[Y]'$ and
  $\enc{u,\gamma'} \models \varphi$. We have to show $\enc{v,\delta'}
  \models \varphi$. Since $\gamma'$ is a restriction of $\gamma$, we
  obtain $\enc{u,\gamma} \models \varphi$. The assumption
  $\enc{u,\gamma} \Dwins\FragmG \enc{v,\delta}$ implies
  $\enc{v,\delta} \models \varphi$. As $\delta'$ is a restriction of
  $\delta$, we conclude $\enc{v,\delta'} \models \varphi$. Thus, $S$
  does not satisfy Spoiler's winning condition.
\item Suppose that Spoiler chooses quantifier $\quant x$ and quest $q$
  in a good configuration $S =
  (\FragmG,\enc{u,\gamma'},\enc{v,\delta'})$. We are looking for a
  suitable reply $r$ for Duplicator. We only demonstrate the case
  $\quant = \exists$, where $q \in \dom{u}$. Let
  $\enc{u,\gamma},\enc{v,\delta}$ be configurations witnessing that $S$ is
  good. Due to $\enc{u,\gamma} \Dwins\FragmG \enc{v,\delta}$, there is
  an $r \in \dom{v}$ such that $\enc{u,\repl\gamma x q}
  \Dwins{\dis\exists x\FragmG} \enc{v,\repl\delta x r}$. Since
  $\enc{u,\repl{\gamma'} x q},\enc{v,\repl{\delta'} x r}$ are
  restrictions of $\enc{u,\repl\gamma x q},\enc{v\repl\delta x r}$,
  the configuration $\move S\exists x q r$ is good.  \qedhere
\end{enumerate}
\end{proof}

\begin{lemma}
\label{lemma:discard_quant}
Let $\Fragm$ be a fragment and $\enc{u,\alpha},\enc{v,\beta}$
$\Varset$-valuations with $\enc{u,\alpha} \Dwins\Fragm
\enc{v,\beta}$. Let $\quant \in \Quantifiers$ and $x \in \Vars
\setminus \Varset$ such that $\dis\quant x\Fragm \not= \emptyset$.
\begin{enumerate}
\item If $\quant \in \{\exists,\forall\}$ is a non-negated quantifier,
  then $\enc{u,\alpha} \Dwins{\dis\quant x\Fragm} \enc{v,\beta}$.
\item If $\quant \in \{\neg\exists,\neg\forall\}$ is a negated
  quantifier, then $\enc{v,\beta} \Dwins{\dis\quant x\Fragm}
  \enc{u,\alpha}$.
\end{enumerate}
\end{lemma}

\begin{proof}
  We only demonstrate the case $\quant = \exists$. First, suppose that
  $\dom{u} \not= \emptyset$. Let $q \in \dom{u}$ be arbitrary. Due to
  the assumption $\enc{u,\alpha} \Dwins\Fragm \enc{v,\beta}$, there
  exists an $r \in \dom{v}$ such that $\enc{u,\repl\alpha x q}
  \Dwins{\dis\exists x\Fragm} \enc{v,\repl\beta x r}$. Since $\alpha$
  and $\beta$ are the restrictions of $\repl\alpha x q$ and
  $\repl\beta x r$ to $\Varset$, respectively,
  Lemma~\ref{lemma:discard_vars} yields the claim.

  Now, assume that $\dom{u} = \emptyset$. If $\dom{v} = \emptyset$ as
  well, then $\enc{u,\alpha} = \enc{v,\beta}$ and the claim is
  trivial. Thus, we further assume $\dom{v} \not= \emptyset$. Let
  $\varphi \in \dis\exists x\Fragm$ be a literal with $\fv(\varphi)
  \subseteq \Varset$ and $\enc{u,\alpha} \models \varphi$. Since $x
  \not\in \fv(\varphi)$, we have $\varphi \in \Fragm$. Thus,
  $\enc{u,\alpha} \Dwins\Fragm \enc{v,\beta}$ implies $\enc{v,\beta}
  \models \varphi$. Hence, $S = (\dis\exists
  x\Fragm,\enc{u,\alpha},\enc{v,\beta})$ does not satisfy Spoiler's
  winning condition. 
\end{proof}

{
\renewcommand*{\thetheorem}{\ref{prop:concatenation}}
\PropositionConcatenation*
\addtocounter{theorem}{-1}
}

\begin{proof}
  Let $u=u_1\dotsm u_k$, $\alpha=\alpha_1\cup\dotsb\cup\alpha_k$,
  $v=v_1\dotsm v_k$, and $\beta=\beta_1\cup\dotsb\cup\beta_k$. We have
  to show $\enc{u,\alpha} \Dwins\Fragm \enc{v,\beta}$. For this
  purpose, we adhere closely to the four steps of constructing a
  winning strategy for Duplicator as described above.
\begin{enumerate}
\item A $\Varset[Y]$-configuration
  $(\FragmG,\enc{u,\alpha'},\enc{v,\beta'})$ is \emph{good} if there
  exist a partition $\Varset[Y] =
  \Varset[Y]_1\uplus\dotsb\uplus\Varset[Y]_k$ and
  $\Varset[Y]_i$-valuations $\enc{u_i,\alpha'_i},\enc{v_i,\beta'_i}$
  for $i \in [1,k]$ such that
  $\alpha'=\alpha'_1\cup\dotsb\cup\alpha'_k$ and
  $\beta'=\alpha'_1\cup\dotsb\cup\alpha'_k$ as well as
  $\enc{u_i,\alpha'_i} \Dwins\FragmG \enc{v_i,\beta'_i}$ for each $i
  \in [1,k]$.

\item Clearly, the configuration $(\Fragm,\enc{u,\alpha},\enc{v,\beta})$
  is good.

\item Let $(\FragmG,\enc{u,\alpha'},\enc{v,\beta'})$ be a good
  $\Varset[Y]$-configuration and $\enc{u_i,\alpha'_i},\enc{v_i,\beta'_i}$
  witnessing $\Varset[Y]_i$-valuations for $i \in [1,k]$. We have to
  show that all literals $\varphi \in \FragmG$ with $\fv(\varphi)
  \subseteq \Varset[Y]$ and $\enc{u,\alpha'} \models \varphi$ also
  satisfy $\enc{v,\beta'} \models \varphi$. The only literals for
  which this is a non-trivial task are $\suc(x,y)$, $\min(x)$ and
  $\max(x)$ as well as their negations. We only demonstrate the
  implication for $\suc(x,y)$. The cases $\min(x)$ and $\max(x)$ as well
  as the negative literals are treated similarly.

  Consider the unique $i,j \in [1,k]$ with $x \in \Varset[Y]_i$ and $y
  \in \Varset[Y]_j$. Since $\enc{u,\alpha'} \models \suc(x,y)$, the
  case $i>j$ cannot occur. The case $i=j$ is trivial. Henceforth, we
  assume $i < j$. We first conclude $\enc{u_i,\alpha'_i} \models
  \max(x)$, $\enc{u_\ell,\alpha'_\ell} \models \mty$ for all $\ell \in
  [i+1,j-1]$, and $\enc{u_j,\alpha'_j} \models \min(y)$. The
  $\suc$-stability of $\FragmG$ implies $\max(x),\min(y),\mty \in
  \FragmG$. Since $\enc{u_\ell,\alpha'_\ell} \Dwins\FragmG
  \enc{v_\ell,\beta'_\ell}$ for all $\ell \in [1,k]$, we obtain
  $\enc{v_i,\beta'_i} \models \max(x)$, $\enc{v_\ell,\beta'_\ell}
  \models \mty$ for all $\ell \in [i+1,j-1]$, and $\enc{v_j,\beta'_j}
  \models \min(y)$. Finally, we conclude $\enc{v,\beta'} \models
  \suc(x,y)$.

\item Suppose that Spoiler chooses quantifier $\quant x$ and quest $q$
  in a good configuration $S=(\FragmG,\enc{u,\alpha'},\enc{v,\beta'})$. We
  are looking for a suitable reply $r$ for Duplicator. We only
  demonstrate the case $\quant = \exists$, where $q \in \dom{u}$. Let
  $\enc{u_i,\alpha'_i},\enc{v_i,\beta'_i}$ for $i \in [1,k]$ be
  valuations witnessing that $S$ is good. There is a unique $i \in
  [1,k]$ such that $q \in \dom{u_i}$. Let $r \in \dom{v_i}$ be such
  that $\enc{u_i,\repl{\alpha'_i} x q} \Dwins{\dis\exists x\FragmG}
  \enc{v_i,\repl{\beta'_i} x r}$. Due to
  Lemmas~\ref{lemma:discard_vars} and~\ref{lemma:discard_quant}, we
  further have $\enc{u_\ell,\alpha''_\ell} \Dwins{\dis\exists
    x\FragmG} \enc{v_\ell,\beta''_\ell}$ for all $\ell \in [1,k]$ with
  $\ell\not=i$, where $\alpha''_\ell$ and $\beta''_\ell$ are the
  restrictions of $\alpha'_\ell$ and $\beta'_\ell$ to
  $\Varset[Y]_\ell\setminus\{x\}$, respectively. Putting $\alpha''_i =
  \repl{\alpha'_i} x q$ and $\beta''_i = \repl{\beta'_i} x r$, we
  obtain
\begin{equation*}
	\repl{\alpha'} x q = \alpha''_1 \cup \dotsb \cup \alpha''_k
	\qquad\text{and}\qquad
	\repl{\beta'} x q = \beta''_1 \cup \dotsb \cup \beta''_k \,.
\end{equation*}
Thus, $\move S\exists x q r$ is a good configuration and hence $r$ is a
suitable reply for Duplicator.\qedhere
\end{enumerate}
\end{proof}

{
\renewcommand*{\thetheorem}{\ref{lemma:exponentiation}}
\LemmaExponentiation*
}

\begin{proof}
  We only show the claim $u^m \Dwins\Fragm v^\theLO$. The proof of the
  other claim is similar. We proceed by induction on $n$.

  \basecase{$n=0$} There are at most six literals $\varphi \in \Fragm$
  without free variables, namely $\top$, $\bot$ and $\mty$ as well as
  their negations. For all of them, one easily shows that $u^m \models
  \varphi$ implies $v^\theLO \models \varphi$. Thus, the configuration
  $(\Fragm,u^m,v^\theLO)$ does not satisfy Spoiler's winning
  condition. Since $\Fragm$ contains no quantified formulae,
  Duplicator immediately wins in this configuration.

  \inductivestep{$n>0$} We assume that $\dom{u^m} = \dom{u}\times
  [1,m]$. Let $(P,\leq_P)$ be a linear ordering of isomorphism type
  $\theLO$ such that $\dom{v^\theLO} = \dom{v} \times P$. For every $p \in P$
  let $\cev\tau_p$ and $\vec\tau_p$ be the order types of the
  suborders of $(P,\leq_P)$ induced by the open intervals
  $(-\infty,p)$ and $(p,+\infty)$, respectively. Then $\theLO =
  \cev\tau_p + 1 + \vec\tau_p$.

  Suppose that Spoiler chooses quantifier $\quant x$ and quest $q$ in
  configuration $(\Fragm,u^m,v^\theLO)$. We are looking for a suitable
  reply $r$ for Duplicator. We only demonstrate the cases $\quant =
  \exists$ and $\quant = \forall$. Recall that
  Lemma~\ref{lemma:discard_quant} implies $u \Dwins{\dis\quant
    x\Fragm} v$ in these two cases.

  \resetcases \case{$\quant = \exists$, where $q = (s,k) \in
    \dom{u^m}$} Since $u \Dwins{\dis\quant x\Fragm} v$, there exists a
  $t \in \dom{v}$ such that $\enc{u,\repl{}x s} \Dwins{\dis\exists
    n\Fragm} \enc{v,\repl{}x t}$.
    Depending on $k$ we will choose a $p
  \in P$ such that
\begin{equation}
\label{eq:conditions_exists}
	u^{k-1} \Dwins{\dis\exists x\Fragm} v^{\smallcev\tau_p}
	\qquad\text{and}\qquad
	u^{m-k} \Dwins{\dis\exists x\Fragm} v^{\smallvec\tau_p} \,.
\end{equation}
By Proposition~\ref{prop:concatenation}, this will imply
$\enc{u^m,\repl{}x{(s,k)}} \Dwins{\dis\exists x\Fragm}
\enc{v^\theLO,\repl{}x{(t,p)}}$, i.e., $r = (t,p)$ is a suitable reply
for Duplicator.

\subcase{$k \leq m-(2^{n}-1)$} Observe that $m-k \geq 2^{n}-1$. We
choose $p \in P$ such that $\cev\tau_p = k-1$ and $\vec\tau_p =
\theLO$. The conditions in Eq.~\eqref{eq:conditions_exists} are met
due to Proposition~\ref{prop:concatenation} and the induction
hypothesis, respectively.

\subcase{$k > m-(2^{n}-1)$} Observe that $k-1 > 2^{n}-1$. We
choose $p \in P$ such that $\cev\tau_p = \theLO$ and $\vec\tau_p =
m-k$. The conditions in Eq.~\eqref{eq:conditions_exists} are met due
to the induction hypothesis and Proposition~\ref{prop:concatenation},
respectively.

\case{$\quant = \forall$, where $q = (t,p) \in \dom{v^\theLO}$} We
proceed similarly to Case~1. There exists an $s \in \dom{u}$ such that
$\enc{u,\repl{}x s} \Dwins{\dis\forall x\Fragm} \enc{v,\repl{}x
  t}$. Depending on $p$ we will choose a $k \in [1,m]$ such that
\begin{equation}
\label{eq:conditions_forall}
	u^{k-1} \Dwins{\dis\forall x\Fragm} v^{\cev\tau_p}
	\qquad\text{and}\qquad
	u^{m-k} \Dwins{\dis\forall x\Fragm} v^{\vec\tau_p} \,.
\end{equation}
Again by Proposition~\ref{prop:concatenation}, this will imply
$\enc{u^m,\repl{}x{(s,k)}} \Dwins{\dis\exists x\Fragm}
\enc{v^\theLO,\repl{}x{(t,p)}}$, i.e., $r = (s,k)$ will be a suitable
reply for Duplicator.

\subcase{$\cev\tau_p$ is finite with $\cev\tau_p<2^{n}-1$ and
  $\vec\tau_p=\theLO$} We choose $k=\cev\tau_p+1 \leq
2^{n}-1$. Thus, $m-k \geq 2^{n}$. The conditions in
Eq.~\eqref{eq:conditions_forall} are met due to
Proposition~\ref{prop:concatenation} and the induction hypothesis.

\subcase{$\cev\tau_p$ is finite with $\cev\tau_p\geq2^{n}-1$ and
  $\vec\tau_p=\theLO$} We choose $k = 2^{n}$. Consequently, $m-k
\geq 2^{n}-1$. The induction hypothesis implies $u^{k-1}
\Dwins{\dis\forall x\Fragm} u^\theLO$ and $u^\theLO \Dwins{\dis\forall
  x\Fragm} v^{\smallcev\tau_p}$. The conditions in
Eq.~\eqref{eq:conditions_forall} are met by transitivity and the
induction hypothesis.

\subcase{$\cev\tau_p = \vec\tau_p = \theLO$} We choose
$k=2^{n}$. Then $m-k \geq 2^{n}-1$. The conditions in
Eq.~\eqref{eq:conditions_forall} both are met due to the induction
hypothesis.

\subcase{$\cev\tau_p=\theLO$ and $\vec\tau_p$ is finite} This case is
symmetric to Cases~2a and~2b.
\end{proof}

\section{Missing Proofs from Section~\ref{sec:EF_terms}}

{
\renewcommand*{\thetheorem}{\ref{cor:term_interpretations}}
\CorollaryTermInterpretations*
}

\begin{proof}
We proceed by induction on the structure of $t$.

\basecase{$t=a$ for some $a\in\Lambda$}
Since $\evalt t m = \evalt t\theLO = a$, the claim is trivial.

\lpara{Inductive step 1}{$t=s_1s_2$ for $\varLO$-terms $s_1,s_2 \in
  \terms$} Due to the induction hypothesis, we have $\evalt{s_i}m
\Indist\Fragm \evalt{s_i}\theLO$ for $i=1,2$. Two applications of
Proposition~\ref{prop:concatenation} yield $\evalt t m = \evalt{s_1}m
\evalt{s_2}m \Indist\Fragm \evalt{s_1}\theLO \evalt{s_2}\theLO =
\evalt t\theLO$.

\lpara{Inductive step 2}{$t=s^\varLO$ for a $\varLO$-term $s \in
  \terms$} Due to the induction hypothesis, we have $\evalt s m
\Indist\Fragm \evalt s\theLO$. Applying
Lemma~\ref{lemma:exponentiation} twice yields $\evalt t m = (\evalt s
m)^m \Indist\Fragm (\evalt s\theLO)^\theLO = \evalt t\theLO$.
\end{proof}

\section{Missing Proofs from Section~\ref{sec:decidability_result}}

The proof of Proposition~\ref{prop:reduction} lacks evidence of the
following claim.

\begin{lemma}
The following identities hold for all $n \geq 1$ and $u,v,w \in \words$:
\begin{gather}
	\tag{Ax1}\label{eq:ax1}
	(u v) w = u (v w) \\
	\tag{Ax2}\label{eq:ax2}
	(u^\theLO)^\theLO = u^\theLO \\
	\tag{Ax3}\label{eq:ax3}
	(u^n)^\theLO = u^\theLO \\
	\tag{Ax4}\label{eq:ax4}
	u^\theLO u^\theLO = u^\theLO \\
	\tag{Ax5}\label{eq:ax5}
	u^\theLO u = u u^\theLO = u^\theLO \\
	\tag{Ax6}\label{eq:ax6}
	(u v)^\theLO u = u (v u)^\theLO
\end{gather}
\end{lemma}

\begin{proof}
  It is a matter of routine to check that the concatenation of words
  is associative, i.e., \eqref{eq:ax1} holds. It is equally simple to
  check that the following two laws of exponentiation hold for any
  word $w \in \words$ and all order types $\sigma$ and $\tau$:
\begin{equation*}
	w^{\sigma + \tau} = w^\sigma w^\tau
	\qquad\text{and}\qquad
	w^{\sigma\cdot\tau} = (w^\sigma)^\tau \,.
\end{equation*}

\para{To (Ax2).}
The three identities below hold for all order types $\sigma$ and $\tau$:
\begin{gather*}
	(\sigma + \tau)\cdot\omega = \sigma + (\tau + \sigma)\cdot\omega \,, \\
	(\sigma + \tau)\cdot\zeta = (\tau + \sigma)\cdot\zeta \,, \\
	(\sigma + \tau)\cdot\omega^* = (\tau + \sigma)\cdot\omega^* + \tau \,.
\end{gather*}
Not that
\begin{gather*}
	(\eta + 1) \cdot \omega = \eta \,, \\
	(\eta + 1) \cdot \zeta = \eta \,, \\
	(1 + \eta) \cdot \omega^* = \eta \,.
\end{gather*}
We conclude
\begin{align*}
	\theLO\cdot\theLO
	&= (\omega + \zeta\cdot\eta + \omega^*) \cdot (\omega + \zeta\cdot\eta + \omega^*) \\
	&= (\omega + \zeta\cdot\eta + \omega^*)\cdot\omega + (\omega + \zeta\cdot\eta + \omega^*)\cdot\zeta\cdot\eta
		+ (\omega + \zeta\cdot\eta + \omega^*)\cdot\omega^* \\
	&= \omega + \zeta\cdot(\eta+1)\cdot\omega + \zeta\cdot(\eta+1)\cdot\zeta\cdot\eta
		+ \zeta\cdot(1+\eta)\cdot\omega^* + \omega^* \\
	&= \omega + \zeta\cdot\eta + \zeta\cdot\eta\cdot\eta + \zeta\cdot\eta + \omega^* \\
	&= \omega + \zeta\cdot(\eta + \eta\cdot\eta + \eta) + \omega^* \\
	&= \theLO \,.
\end{align*}
Thus,
\begin{equation*}
	(u^\theLO)^\theLO = u^{\theLO\cdot\theLO} = u^\theLO \,.
\end{equation*}

\para{To (Ax3).}
We use the following tree identities:
\begin{align*}
	n\cdot\omega &= \omega \,, &
	n\cdot\zeta &= \zeta \,, &
	n\cdot\omega^* &= \omega^* \,.
\end{align*}
Consequently,
\begin{equation*}
	n\cdot\theLO
	= n\cdot(\omega + \zeta\cdot\eta + \omega^*)
	= n\cdot\omega + n\cdot\zeta\cdot\eta + n\cdot\omega^*
	= \omega + \zeta\cdot\eta + \omega^*
	= \theLO
\end{equation*}
and hence
\begin{equation*}
	(u^n)^\theLO = u^{n\cdot\theLO} = u^\theLO \,.
\end{equation*}

\para{To (Ax4).}
We have
\begin{equation*}
	\theLO + \theLO
	= \omega + \zeta\cdot\eta + \omega^* + \omega + \zeta\cdot\eta + \omega^*
	= \omega + \zeta\cdot(\eta+1+\eta) + \omega^*
	= \theLO
\end{equation*}
and therefore
\begin{equation*}
	u^\theLO u^\theLO = u^{\theLO + \theLO} = u^\theLO \,.
\end{equation*}

\para{To (Ax5).}
We have
\begin{equation*}
	\theLO + 1 = \omega + \zeta\cdot\eta + \omega^* + 1 = \theLO
	\qquad\text{and}\qquad
	1 + \theLO = 1 + \omega + \zeta\cdot\eta + \omega^* = \theLO \,.
\end{equation*}
Consequently,
\begin{equation*}
	u^\theLO u = u^{\theLO + 1} = u^\theLO
	\qquad\text{and}\qquad
	u u^\theLO = u^{1 + \theLO} = u^\theLO \,.
\end{equation*}

\para{To (Ax6).}
For all $a,b \in \Lambda$, the following identities are easy to check:
\begin{align*}
	(a b)^\omega &= a (b a)^\omega \,, &
	(a b)^\zeta &= (b a)^\zeta \,, &
	(b a)^{\omega^*} &= (a b)^{\omega^*} a \,.
\end{align*}
Replacing in the words occurring in these identities every
$a$-labeled position by the word $u$ and every $b$-labeled position
by the word $v$, yields the following identities:
\begin{align*}
	(u v)^\omega &= u (v u)^\omega \,, &
	(u v)^\zeta &= (v u)^\zeta \,, &
	(v u)^{\omega^*} &= (u v)^{\omega^*} u \,.
\end{align*}
Thus,
\begin{equation*}
	(u v)^\theLO u
	= (u v)^\omega \bigl((u v)^\zeta\bigr)^\eta (u v)^{\omega^*} u
	= u (v u)^\omega \bigl((v u)^\zeta\bigr)^\eta (v u)^{\omega^*}
	= u (v u)^\theLO \,.
	\qedhere
\end{equation*}
\end{proof}

\end{appendix}

\end{document}